\documentclass[aps,pra,reprint,10pt,superscriptaddress,nofootinbib]{revtex4-2}
\usepackage{amsmath,amssymb,amsfonts}
\usepackage{bm}
\usepackage{mathrsfs}
\usepackage{braket}
\usepackage{graphicx}
\usepackage[colorlinks=true,citecolor=blue,urlcolor=blue]{hyperref}
\usepackage{float}
\usepackage{tikz}

\usetikzlibrary{arrows.meta}
\usepackage{microtype}
\usepackage{amsthm}
\usepackage[normalem]{ulem} 

\newtheorem{theorem}{Theorem}
\newtheorem{lemma}{Lemma}
\newtheorem{corollary}{Corollary}

\definecolor{indiagreen}{rgb}{0.07, 0.53, 0.03}
\definecolor{teal}{rgb}{0.0, 0.53, 0.53}

\begin{document}
\title{Robust Genuine Multipartite Entanglement in Two Walker Quantum Walks}
\author{Sandipan Hazra}
\email{sandipan.hazra.phy@gmail.com}
\author{Tamoghna Das}
\affiliation{Department of Physics, Indian Institute of Technology Kharagpur, Kharagpur 721302, India}

\author{Sougato Bose}
\affiliation{Department of Physics and Astronomy, University College London, London WC1E 6BT, United Kingdom}
\author{Sonjoy Majumder}
\affiliation{Department of Physics, Indian Institute of Technology Kharagpur, Kharagpur 721302, India}
\date{\today}

\begin{abstract}
Discrete-time quantum walks provide a versatile framework for investigating the generation, redistribution, and transport of quantum correlations in composite quantum systems. Here, we study the dynamics of bipartite and genuine multipartite entanglement in a two-walker discrete-time quantum walk on a one-dimensional lattice. By employing logarithmic negativity and the generalized geometric measure (GGM), we systematically characterize the redistribution of bipartite entanglement among different subsystem partitions and the emergence of genuine multipartite entanglement involving the two coin and two position degrees of freedom. We show that the entanglement dynamics are strongly influenced by the lattice topology. The open-boundary regime exhibits a monotonic redistribution of quantum correlations, whereas the closed-boundary regime gives rise to pronounced oscillatory behavior due to boundary-induced interference and recurrent wave-packet overlap. In the open-boundary regime, the GGM rapidly approaches its theoretical maximum value of $1/2$ and remains largely insensitive to the choice of the initial Bell state as well as to continuous variations of the local coin operator over a broad parameter range, except near the Pauli-$X$ coin. These results demonstrate that maximal genuine multipartite entanglement generation is a robust and generic feature of open-boundary two-walker discrete-time quantum walks, establishing them as promising platforms for engineering multipartite quantum correlations in quantum information processing and quantum simulation.
\end{abstract}

\maketitle
\section{Introduction}

Quantum correlations \cite{RevModPhys.81.865,RevModPhys.84.1655} are fundamental resources for quantum information science and underpin a wide range of quantum technologies, including quantum communication \cite{PhysRevLett.69.2881}, quantum computation \cite{Bennett2000}, quantum cryptography \cite{RevModPhys.74.145}, quantum sensing \cite{RevModPhys.89.035002}, and quantum simulation~\cite{RevModPhys.86.153}. Consequently, understanding the generation, evolution, and preservation of quantum correlations in composite quantum systems is a central problem in the foundations of quantum mechanics and quantum information science. Quantum walks~\cite{PhysRevA.48.1687,Farhi1998,Kempe01072003} provide a fundamental framework for studying coherent quantum transport and interference and have emerged as versatile tools in quantum information processing, quantum simulation, and the investigation of quantum dynamics. Two principal formulations of quantum walks have been extensively investigated in the literature: Aharonov et al. introduced discrete-time quantum walks (DTQWs)~\cite{PhysRevA.48.1687,Kempe01072003} and their continuous-time version, known as continuous-time quantum walks (CTQWs)~\cite{Farhi1998,Childs2003} was formulated within a few years. Both models have played central roles in the development of quantum algorithms~\cite{Childs2003,PhysRevA.67.052307,2462630,PhysRevLett.102.180501,Lovett2012}, universal quantum computation~\cite{PhysRevLett.102.180501,PhysRevA.81.042330}, quantum state transfer and transport~\cite{PhysRevA.83.062315,Stefanak2017,Ikken2025}, and quantum simulation~\cite{Mallick_2019, Arnault_2020}.\\
\indent In this work, we focus on DTQWs, which are characterized by the interplay between an internal coin degree of freedom and conditional position shifts. The presence of the internal coin degree of freedom gives rise to rich interference effects, ballistic spread, and highly nonclassical transport properties~\cite{nayak2000quantumwalkline, VenegasAndraca2012}. Experimentally, quantum walks have been realized on a variety of physical platforms, including trapped ions~\cite{PhysRevLett.103.090504}, neutral atoms in optical lattices~\cite{doi:10.1126/science.1174436}, photonic systems~\cite{PhysRevLett.104.153602,doi:10.1126/science.1193515,PhysRevLett.100.170506}, superconducting circuits~\cite{PhysRevX.7.031023,5lm1-2kpk}, and nuclear magnetic resonance systems~\cite{PhysRevA.72.062317}.\\
\indent Entanglement generation in DTQWs has been extensively investigated for single-particle quantum walks, where coin-position entanglement emerges naturally during the evolution~\cite{PhysRevA.73.042302,Carneiro_2005,Goyal_2010,PhysRevA.105.042216}. Multi-particle quantum walks exhibit significantly richer correlation structures owing to the coexistence of inter and intra-particle entanglement channels~\cite{PhysRevA.74.042304,PhysRevA.75.032351,Rodriguez2015,Carson2015,PhysRevA.100.042110}. In particular, two-particle quantum walks provide a minimal yet versatile setting for investigating entanglement transfer, nonlocal quantum correlations, and interference-induced dynamics.\\
\indent Although these studies have significantly advanced the understanding of bipartite quantum correlations in quantum walks, a systematic characterization of bipartite entanglement using logarithmic negativity\cite{PhysRevA.65.032314} together with genuine multipartite entanglement in two-walker discrete-time quantum walks remains largely unexplored. Very few previous works studied entanglement in the quantum walk formalism using logarithmic negativity~\cite{PhysRevE.108.024139,PhysRevA.110.052428}. On the other hand, genuine multipartite entanglement captures quantum correlations shared collectively among all constituent subsystems, and therefore cannot be inferred solely from bipartite entanglement measures. Among the available quantifiers, the generalized geometric measure (GGM) \cite{PhysRevA.81.012308}, provides an operationally meaningful and computationally efficient measure of genuine multipartite entanglement for arbitrary pure multipartite states.\\
\indent An equally important question concerns the robustness of multipartite entanglement generated during the quantum walk. In particular, it remains unclear whether the generation of genuine multipartite entanglement is intrinsically tied to a specific initial Bell state or coin operator, or whether it represents a generic feature of the quantum-walk dynamics. Furthermore, the influence of the lattice topology, the initial coin entanglement, and the local coin operation on the generation and stability of genuine multipartite entanglement has not been systematically investigated.

\indent In this work, we address these intriguing questions by investigating first the dynamics of two non-interacting walkers on a one-dimensional lattice using logarithmic negativity and the generalized geometric measure (GGM) to characterize bipartite and genuine multipartite entanglement (GME), respectively. We analyze the redistribution of bipartite entanglement among all physically relevant subsystem partitions and 
qualitatively understand the dependence of the entanglement dynamics on the evolution of walkers in the open- or closed-boundary regime. We further demonstrate that, within the class of initial Bell states and local coin operators considered here, the generated multipartite entanglement is remarkably robust and rapidly approaches its theoretical maximum in the open-boundary regime. These results may establish maximal genuine multipartite entanglement generation as a robust and generic feature of two-walker discrete-time quantum walks, thereby establishing DTQWs as promising platforms for engineering multipartite quantum correlations.

The remainder of the paper is organized as follows. Section~II introduces
the two-walker discrete-time quantum-walk model, while Sec.~III presents
the entanglement measures used in our analysis. Section~IV examines the
entanglement dynamics in the open- and closed-boundary regimes. Section~V
investigates the robustness of genuine multipartite entanglement against
variations in the initial coin state and local coin operator. Finally,
Sec.~VI summarizes the main findings and their implications. Analytical
derivations are provided in the Appendices.
\section{Theory}
\subsection{Model and Formalism}
We consider a discrete-time quantum walk (DTQW) of two non-interacting walkers (A and B) on a finite one-dimensional lattice with $N$ vertices. Each walker possesses two degrees of freedom: one is internal, related to the coin, and the other is external, governed by the shifting position. The total Hilbert space can be expressed as
\begin{equation}
\mathcal{H} = \mathcal{H}_{C_A} \otimes \mathcal{H}_{P_A}\otimes \mathcal{H}_{C_B} \otimes \mathcal{H}_{P_B},
\end{equation}
where $\mathcal{H}_{C_i} \cong \mathbb{C}^2$ is spanned by the orthonormal coin basis states $\{|0\rangle,|1\rangle\}$, and $\mathcal{H}_{P_i} \cong \mathbb{C}^N$ is spanned by the position basis states $\{|x\rangle\}$ with $x\in\{0,1,\ldots,N-1\}$, for $i\in\{A,B\}$. Here, $|0\rangle$ and $|1\rangle$ denote the two basis states of the coin space of each walker, while $|x\rangle$ represents the walker localized at the lattice site labeled by the integer $x$. 

The DTQW evolution is governed by a coin operator $O_{C_{A(B)}}$ followed by a conditional shift operator $S_{A(B)}$. The total coin operator functions locally on each walker and is expressed as
\begin{equation}
C_{AB} = (O_{C_A} \otimes O_{C_B}) \otimes (I_{P_A} \otimes I_{P_B}),
\label{def_C}
\end{equation}
In this study, we use the Hadamard operator as the local coin, as it produces an unbiased quantum walk by allocating equal amplitudes to the left- and right-moving states. Due to its simplicity and balanced dynamics, the Hadamard coin has been extensively employed in the discrete-time quantum-walk literature \cite{10.1145/380752.380757, 7rwg-lhpv}. So,
\begin{equation*}
    O_{C_A}=H=
\frac{1}{\sqrt{2}}
\begin{pmatrix}
1 & 1\\
1 & -1
\end{pmatrix}=O_{C_B} .
\end{equation*}
\\Similarly the joint shift operator is defined as
\begin{equation}
S_{AB} = S_A \otimes S_B,
\label{def_S}
\end{equation}
where each conditional single-party shift operator is
\begin{equation}
\begin{aligned}
S_{A(B)}
=
\sum_x \Big(
&\ket{0}\bra{0}_{C_{A(B)}}
\otimes
\ket{x + 1}\bra{x}_{P_{A(B)}}
\\
+ 
&\ket{1}\bra{1}_{C_{A(B)}}
\otimes
\ket{x-1}\bra{x}_{P_{A(B)}}
\Big).
\end{aligned}
\label{Eq:S_A}
\end{equation}
At each discrete step, $t \in \mathbb{N}$, the joint quantum state is governed by $\ket{\Psi(t)}_{AB} = U_{AB} \ket{\Psi(t-1)}_{AB}$, where the unitary evolution operator is
\begin{equation}
U_{AB} = S_{AB} \hspace{0.02in}C_{AB}.
\label{def_U}
\end{equation}
It is important to note that the joint unitary operator $U_{AB}$ can be expressed as a tensor product of two local unitary operators, 
\begin{equation}
    U_{AB}=U_A \otimes U_B, 
    \label{eq:U_AB}
\end{equation}
where $U_A=S_A(O_{C_A}\otimes I_{P_A})$ and $U_B=S_B(O_{C_B}\otimes I_{P_B})$.\\
The dynamics of quantum walk are driven by repeated application of the joint unitary operator $U_{AB}$. Since the evolution operator has no explicit time dependence, the same unitary operator acts at every discrete time step. Consequently, after discrete steps $t$, the final state evolves as
\begin{equation}
\ket{\Psi(t)}_{AB} = (U_{AB})^t \ket{\Psi(0)}_{AB}, \label{eq:finalPsit}
\end{equation}
where $(U_{AB})^t$ represents the successive application of the same unitary operator $U_{AB}$ exactly $t$ times, as $t\ \in \mathbb{N}$.
\subsection{Boundary Conditions}
Suppose at $t = 0$, the walker starts
 from a localized position state $|x\rangle$, and after $t = T$, the walker can occupy a maximum lattice distance $T$ considering unit lattice spacing from either side of its initial position $|x\rangle$, resulting in the total spatial extent explored by the walker to be at most $2T+1$ lattice sites.\\
\indent In the present work, we consider a finite lattice  of $N$ sites with site indices
$x \in \{0,1,2,\ldots,N-1\}$ and distinguish between two physically different regimes:
\subsubsection{Open-boundary regime ($N>2T$)} In this case, the lattice size exceeds the maximum spread of the walkers during the evolution. The walkers never reach the lattice boundaries within $T$ steps, and the dynamics is effectively identical to that on an infinite one-dimensional lattice.
\subsubsection{Closed-boundary regime ($N<2T$)}Here, the lattice size is smaller than the maximum spatial extent of the walk. The walkers can traverse the entire lattice and encounter the boundaries during the evolution. As a result, finite-size effects become important and multiple paths interfere after wrapping around the lattice. For this case,  the symbols $``+"$ and $``-"$ of Eq. \eqref{Eq:S_A} will reduce to the ``modulo N plus" and ``modulo N minus" i.e., $ + \equiv \oplus_N \hspace{0.1in}(- \equiv \ominus_N)$\footnote{The symbols $\oplus_N$ and $\ominus_N$ denote addition and subtraction modulo $N$, respectively, i.e., $x\oplus_N1=(x+1)\bmod N$ and  $x\ominus_N1=(x-1)\bmod N$.}, thereby implementing the periodic boundary condition on the lattice.
\\Throughout the work, whenever we consider the closed-boundary regime, we use the fact that 
\begin{equation}
     \ket{x\oplus_N N}\equiv\ket{x},
\end{equation}
and all position coordinates are understood modulo $N$, as $N$ being the total number of lattice sites considered. In this case, the walkers re-enter the lattice from the opposite side after crossing a boundary, leading to additional interference effects, which is absent in the open-boundary regime.
\begin{figure}[ht]
\centering
\begin{tikzpicture}[scale=1.0]
\draw[thick] (90:2) arc[start angle=90,end angle=210,radius=2];
\draw[thick,dotted] (210:2) arc[start angle=210,end angle=270,radius=2];
\draw[thick,dotted] (270:2) arc[start angle=270,end angle=330,radius=2];
\draw[thick] (330:2) arc[start angle=330,end angle=450,radius=2];
\node[circle,fill=black,inner sep=1.5pt,label=above:{$x=0$}] (n0) at (90:2) {};
\node[circle,fill=black,inner sep=1.5pt,label=left:{$x=1$}] (n1) at (150:2) {};
\node[circle,fill=black,inner sep=1.5pt,label=below left:{$x=2$}] (n2) at (210:2) {};
\node[circle,fill=black,inner sep=1.5pt,label=below right:{$x=N-2$}] (n3) at (330:2) {};
\node[circle,fill=black,inner sep=1.5pt,label=right:{$x=N-1$}] (n4) at (30:2) {};
\node[circle,draw,fill=white,thick,inner sep=1.5pt,
      label=below:{$x=\lfloor N/2 \rfloor$}] (start) at (270:2) {};
\draw[-Stealth,thick]
(-1.6,-2.6) -- (-0.18,-2.05);
\node[align=center] at (-2.8,-2.8)
{Initial position\\of both walkers};
\node at (0,0) {$N$ sites};
\end{tikzpicture}
\caption{
One-dimensional lattice with periodic boundary conditions. The lattice sites are labelled by
$x=0,1,2,\ldots,N-1$, with the identification
$\ket{x+N}\equiv\ket{x}$.
The walkers are initially localized at the central lattice site
$x=\lfloor N/2\rfloor$, where $\lfloor\cdot\rfloor$ denotes the floor function. The periodic identification makes the sites $x=0$ and $x=N-1$ nearest neighbours, resulting in a closed lattice topology.
}
\label{fig:ring}
\end{figure}
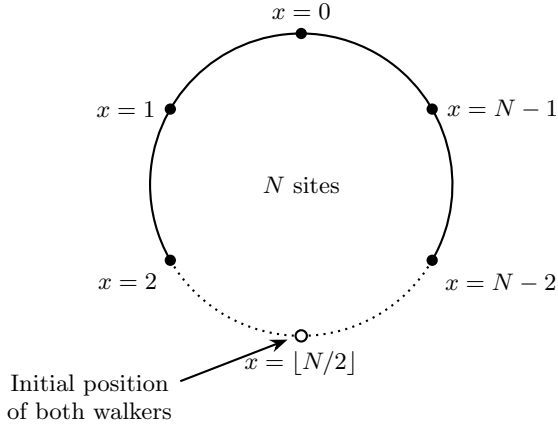

\subsection{Initial States}
At time $t = 0$, we now posit that both walkers are localized at the central point of a one-dimensional lattice illustrated in Fig. \ref{fig:ring}, indicated by an arrow. Hence, the initial joint quantum state, will take the form of
\begin{equation}
\ket{\Psi(0)}_{AB}
=
\ket{\psi}_{C_AC_B}
\otimes
\ket{x_A,x_B:x_A=x_B=\lfloor N/2\rfloor}_{P_AP_B},
\end{equation}
where $\ket{\psi}_{C_AC_B}$ denotes the initial two-coin state and $\lfloor N/2 \rfloor$ is the greatest integer less than or equal to $N/2$.\\
\indent We have considered two fundamentally different choices of the initial coin states, namely, the separable states and the maximally entangled Bell states.
We consider two separable initial coin states,
\begin{align}
\ket{\psi_{1}}_{C_AC_B} &= \ket{01}_{C_AC_B}, \\
\ket{\psi_{2}}_{C_AC_B} &= \ket{0+}_{C_AC_B},
\end{align}
where
\[
\ket{+}=\frac{1}{\sqrt{2}}(\ket{0}+\ket{1}),
\]
together with the four maximally entangled Bell states,
\begin{align}
\ket{\psi_{3}}_{C_AC_B}
&=
\ket{\Psi^-}_{C_AC_B}
=
\frac{1}{\sqrt{2}}
\left(
\ket{01}
-
\ket{10}
\right)_{C_AC_B},
\\
\ket{\psi_{4}}_{C_AC_B}
&=
\ket{\Psi^+}_{C_AC_B}
=
\frac{1}{\sqrt{2}}
\left(
\ket{01}
+
\ket{10}
\right)_{C_AC_B},
\\
\ket{\psi_{5}}_{C_AC_B}
&=
\ket{\Phi^-}_{C_AC_B}
=
\frac{1}{\sqrt{2}}
\left(
\ket{00}
-
\ket{11}
\right)_{C_AC_B},
\\
\ket{\psi_{6}}_{C_AC_B}
&=
\ket{\Phi^+}_{C_AC_B}
=
\frac{1}{\sqrt{2}}
\left(
\ket{00}
+
\ket{11}
\right)_{C_AC_B}.
\end{align}
For both qualitatively different initial coin states, we will study the evolution of the bipartite as well as multipartite  entanglement developed among the subsystems during the DTQW. However, before proceeding, in the next section we briefly discuss the entanglement measures we used to quantify the entanglement among the subsystems.

\section{Measures of Bipartite and Multipartite Entanglement}
To characterize the quantum correlations generated during the discrete-time quantum walk of two walkers, we use the following two measures. 
We use logarithmic negativity ($E_{\mathcal N}$) ~\cite{Peres1996, Horodecki1996} to quantify bipartite entanglement among sub-systems. 
\begin{equation}
E_{\mathcal N}(\rho_{AB})
=\log_2
\left(
\sum_i |\lambda_i|
\right),
\label{eq}
\end{equation}
where $\lambda_i$ are the eigenvalues of partial transpose of $\rho_{AB}$ with respect to subsystem $A$ or $B$.

To quantify the genuinely multipartite entanglement of the total system, we will use Generalized Geometric Measure (GGM) \cite{PhysRevA.81.012308}.

For an $M$-party pure state $\ket{\Psi_M}$, the generalized geometric measure (GGM) is defined as
\begin{equation}
GGM(\ket{\Psi_M})
=
1-
\max_{\ket{\Phi_M}}
\left|
\langle \Phi_M | \Psi_M \rangle
\right|^2,
\label{eq:GGM_def}
\end{equation}
where maximization is taken over all possible pure states $\ket{\Phi_M}$ that are not genuinely multipartite entangled, i.e., states that are product at least in one bipartition.\\
One can easily verify that the above expression of GGM reduces to a much more computationally convenient expression, given by
\begin{equation}
GGM(\ket{\Psi_M})=1-\max_{\mathcal A:\mathcal B}
\left\{\lambda_{\mathcal A:\mathcal B}^{2}, \mathcal A \cap \mathcal B = \emptyset \text{ and } \mathcal A \cup \mathcal B = M \right\},
\label{eq:GGM}
\end{equation}
where $\lambda_{\mathcal A:\mathcal B}$ denotes the largest Schmidt coefficient across the bipartition $\mathcal A:\mathcal B$, with $\mathcal A$ and $\mathcal B$ forming a partition of the possible subsystems. The maximization is performed over all nontrivial bipartitions of the system.

For the present two-walker DTQW, the total state evolves in the multipartite Hilbert space
$\mathcal{H} = \mathcal{H}_{C_A} \otimes \mathcal{H}_{P_A}\otimes \mathcal{H}_{C_B} \otimes \mathcal{H}_{P_B}$ where the maximization is carried out over all nontrivial bipartitions of these four subsystems.
The GGM therefore quantifies the genuine multipartite entanglement among the four degrees of freedom shared among the coin and position spaces of the two walkers.\\
\indent The simultaneous use of logarithmic negativity and GGM enables us to distinguish between bipartite entanglement distribution among the subsystems and the generation of genuine multipartite quantum correlations during the quantum walk dynamics.

\section{Entanglement Dynamics}

In this section, we discuss the generation of bipartite and multipartite entanglement of various time-evolved DTQW cases. For example, when the lattice sites of the walkers are unconstrained, i.e., in the open-boundary regime and when it is constrained to follow closed-boundary conditions for various choices of the coin-coin input states. 

\subsection{Open-boundary regime ($N>2T$)}
Here we consider the situation where both walkers starting from a middle point $\lfloor N/2 \rfloor$ with $N  = 2T +1$ move freely in either direction. 
We first focus on the condition where both the initial state of the composite coin is unentangled or separable (considering the ``separable" states here are pure product states).
\subsubsection{Separable Coin-Coin State} 
\label{subsection-separable}
The entanglement dynamics of a two-walker quantum walk with an initially separable coin-coin state has been investigated to some extent in~\cite{Singh_2019}. Since the evolution is governed by the joint unitary $U_{AB}$, given in 
 Eq. \eqref{eq:finalPsit}, which is in product form for the $A: B$ bipartition, it is trivial to prove that there will be no bipartite entanglement in the  $C_A:C_B $, $P_A:P_B$, $C_AP_A:C_BP_B$, $C_A:P_B$ and $C_B:P_A$ bipartitions. However, the entanglement can develop for the $C_AC_B:P_AP_B$, $C_A:P_A$ and $C_B:P_B$ cuts. This follows directly from the local structure of the evolution operator, $U_{AB}=U_A\otimes U_B$, together with the initially separable state
$\ket{\Psi(0)}_{AB}=\ket{\Psi(0)}_A\otimes\ket{\Psi(0)}_B.$

Consequently, the evolved state retains the factorized form
$\ket{\Psi(t)}_{AB}=U_A^t\ket{\Psi(0)}_A \otimes U_B^t\ket{\Psi(0)}_B,$
which precludes the generation of entanglement between the two walkers. Moreover, since $|\Psi(t)\rangle$, is always a product in $C_AP_A:C_BP_B$, there will be no genuine multipartite entanglement. 
\begin{figure}[t]
    \subsubsection*{For N=51}
    \includegraphics[width=1\linewidth]{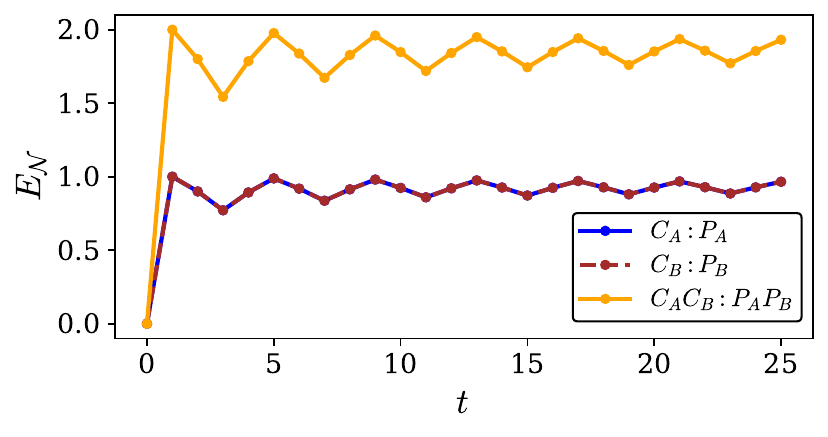}  
    \caption{ (Color online.)
    Time evolution of bipartite entanglement (in ebit unit) in various bipartitions has been plotted as a function of number of steps $t$ ( dimensionless), in the two-walker DTQW, when the initial coin: coin quantum state is a product state $\ket{\psi_{1}}_{C_AC_B}=\ket{01}_{C_AC_B}$. The entanglement has been calculated with the help of logarithmic negativity in the bipartitions $C_A: P_A$ (by the solid blue line) and $C_B: P_B$ (by the dashed brown line) $C_AC_B: P_AP_B$ (by the solid orange curve). Here we consider the open-boundary regime with $ N=51$ lattice points and a maximum allowed step count of $T=25$.}
    \label{fig:N101_T50_sep}
\end{figure}

Fig.~\ref{fig:N101_T50_sep} demonstrates the entanglement variations of the non-trivial bipartitions, where the entanglement is developed from 
 the initial coin state, $\ket{\psi_{1}}_{C_AC_B}=\ket{01}_{C_AC_B}$, due to the DTQW.  
 Here, the maximum lattice size is chosen to be $N=51$ satisfying $N>2T$. Therefore, walkers do not encounter the boundaries when $t < T = 25$. This dynamics effectively corresponds to the evolution of entanglement on an unbounded one-dimensional lattice. 

 Note that, in principle, the $T$ can be arbitrarily large, but for the ease of our numerical calculation, we impose the constraint of maximum lattice size and maximum time steps.

The behavior of the collective coin-position entanglement, $E_\mathcal{N}(C_AC_B: P_AP_B)$, has been plotted in Fig.~\ref{fig:N101_T50_sep}, by the 
solid orange line as a function of discrete time $t$. 
It first rapidly increases, reaches its maximum at $t = 1$, and then oscillates with a small amplitude. The range of oscillations decreases over time and asymptotically approaches the maximum value. 

\begin{figure}[t]
    \subsubsection*{For N=51}  
    \includegraphics[width=1\linewidth]{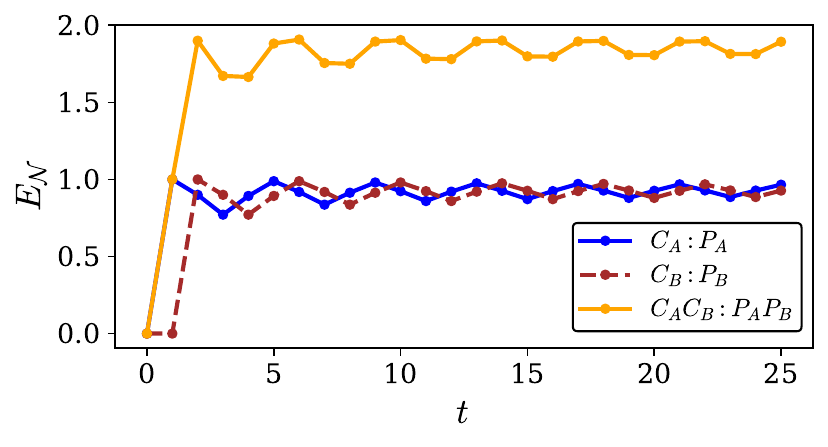}
    \caption{(Color online.)
    Time evolution of bipartite entanglement (in ebit unit) in various bipartitions has been plotted as a function of number of steps $t$ ( dimensionless), in the two-walker DTQW, when the initial coin: coin quantum state is a product state $\ket{\psi_{2}}_{C_AC_B}=\ket{0+}_{C_AC_B}$. The entanglement has been calculated with the help of logarithmic negativity in the bipartitions $C_A: P_A$ (by the solid blue line) and $C_B: P_B$ (by the dashed brown line) $C_AC_B: P_AP_B$ (by the solid orange curve). Here, we consider the open-boundary regime with $N=51$ lattice points and a maximum allowed step count of $T=25$.}
    \label{fig:N51_T25_sep_0+}
\end{figure}

This behaviour of entanglement demonstrates that the quantum walk efficiently generates correlations between the collective coins and their positions while preserving the separability of the two walkers.

The single-particle behavior has also been investigated considering the reduced density matrix $\rho_{A(B)} = \text{tr}_{B(A)}\left( \ket{\Psi(t)}\bra{\Psi(t)}_{AB}\right)$ for walker $A$ or $B$. Although the initial coin states of $A$ and $B$ are mutually orthogonal, our numerical calculations find that the bipartite entanglement between the coin and position of $A$ walker, ($E_\mathcal{N}(C_A:P_A)$, is identical with the same of $B$, $E_\mathcal{N}(C_B:P_B)$), shown in Fig.~\ref{fig:N101_T50_sep} by solid blue and dashed brown lines, respectively. 
This is due to the parity symmetry of the independent operations of $A$ and $B$ about the lattice point $N/2$. 
Moreover, one can easily check that for the pure product state the collective coin:position entanglement is twice that of a single walker, as $E_\mathcal{N}(C_AC_B: P_AP_B) = E_\mathcal{N}(C_A:P_A) + E_\mathcal{N}(C_B:P_B) = 2E_\mathcal{N}(C_A:P_A)$. The first equality emerges from the additivity property of Logarithmic negativity for product states, while the second one is due to the identical values obtained for the bipartite entanglement between the subsystems of individual walkers A and B.

We have also investigated the dynamics for another initial coin state in product form, $\ket{\psi_2}_{C_AC_B}=\ket{0+}_{C_AC_B}$, where $\ket+=\frac{1}{\sqrt2}(\ket{0}+\ket1)$. The dynamics is qualitatively mostly similar, with a  difference in the inequal profiles of $E_\mathcal{N}(C_A:P_A)$ and $E_\mathcal{N}(C_B:P_B)$ , as plotted in Fig. \ref{fig:N51_T25_sep_0+}.

Hence, we can conclude that entanglement is generated only between the coin and position degrees of freedom belonging to the same walker. The absence of cross correlations reflects the fact that the two walkers evolve independently throughout the quantum walk.

\subsubsection{Maximally Entangled Coin-Coin State}
\begin{figure*}[ht]
\label{GGM_gen}
\subsubsection*{For $N=51$}
\centering

\begin{minipage}[ht]{0.45\linewidth}
    \centering
    \includegraphics[width=\linewidth]{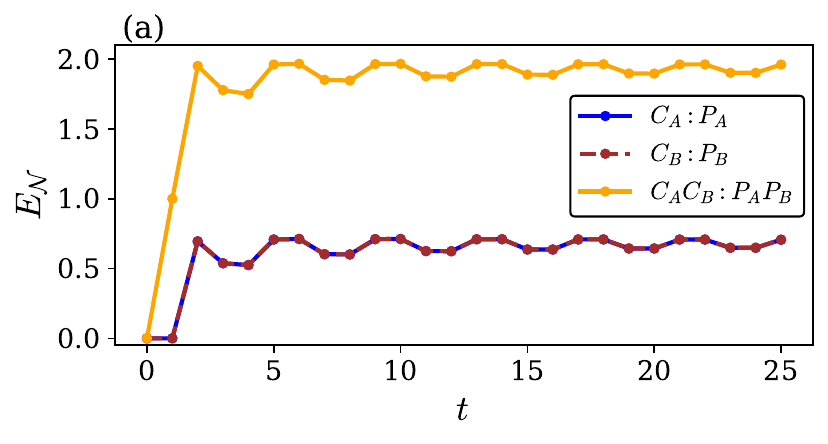}
    \includegraphics[width=\linewidth]{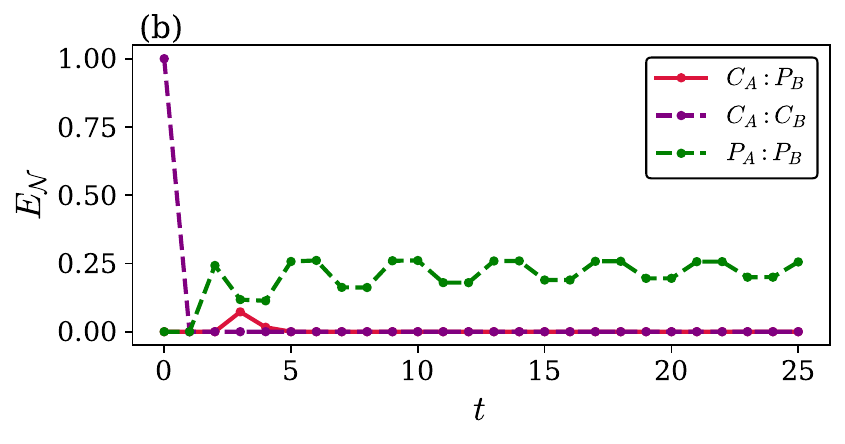}
\end{minipage}
\hfill
\raisebox{-0.1cm}{%
\begin{minipage}[ht]{0.53\linewidth}
    \centering
    \includegraphics[width=\linewidth]{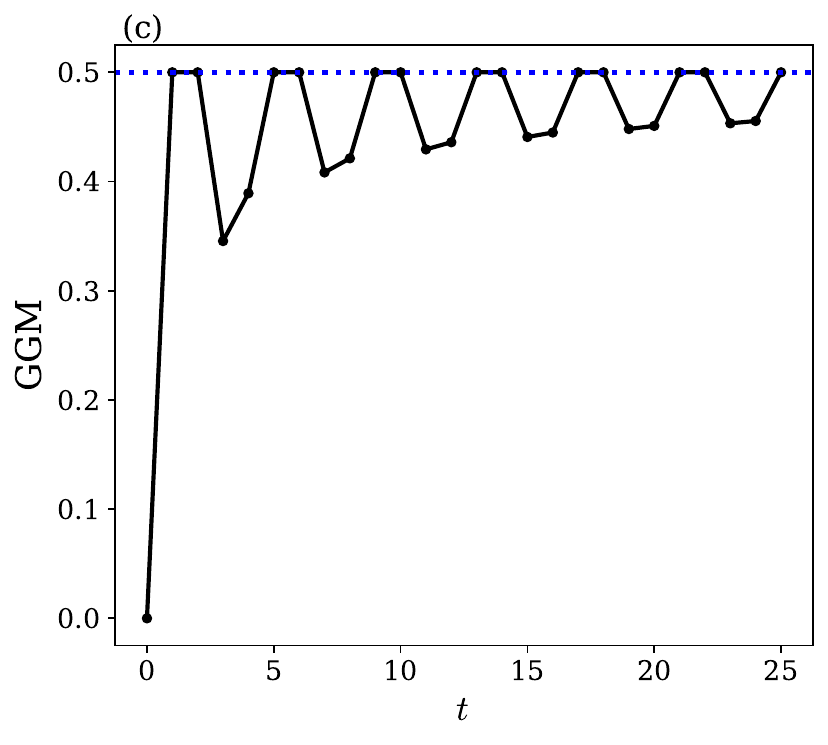}
\end{minipage}}
\caption{(Color online.)
    Time evolution of bipartite entanglement (in ebit unit) in various bipartitions and GGM (dimensionless) has been plotted as a function of number of steps $t$ (dimensionless), in the two-walker DTQW, when the initial coin: coin quantum state is $\ket{\psi_{3}}_{C_AC_B}=\frac{1}{\sqrt{2}} \left( \ket{01}-\ket{10} \right)_{C_AC_B}$ The entanglement has been calculated with the help of logarithmic negativity in the bipartitions (a) $C_A:P_A$, $C_B:P_B$ and $C_AC_B:P_AP_B$, and (b) $C_A:P_B$ , $C_A:C_B$ and $P_A: P_B$ whereas in (c) the generalized geometric measure (GGM) of $\ket{\psi(t)}_{AB}$ has been plotted. Here, we consider the open-boundary regime with $N=51$ lattice points and a maximum allowed step count of $T=25$.}
\label{fig:N51_T25_entangled}
\end{figure*}
In this section, we are going to investigate the generation of bipartite and multipartite entanglement and its dynamical behavior, when the initial coin state is chosen to be the Bell singlet state, i.e., 
$$\ket{\psi_{3}}_{C_AC_B}
=\frac{1}{\sqrt{2}}
\left(
\ket{01}-\ket{10}
\right)_{C_AC_B}.$$

In contrast to the separable initial state, the presence of initial coin-coin entanglement leads to the generation of nontrivial multipartite correlations in different bipartitions involving both the coin and position degrees of freedom.
In the case of single walker subsystem, the entanglement between the coin:position bipartition ($E_\mathcal{N}(C_A:P_A)$ and $E_\mathcal{N}(C_B:P_B)$)  has been plotted by solid blue line and dashed brown line in Fig.~\ref{fig:N51_T25_entangled}(a). 
Note that the time evolved single walker reduced density matrices $\rho_x, ~~x \in \{A,B\}$ are identical to each other, as for both the walker initial coin state is a maximally mixed one, resulting 
$E_\mathcal{N}(C_B:P_B) = E_\mathcal{N}(C_A:P_A)$. The entanglement rapidly grows initially from $0$ to the value of around $0.69$ ebits and subsequently oscillates near to the steady-state value, emphasizing our claim of redistribution of coin:coin entanglement among all the remaining subsystems. Simultaneously, the collective coin-position entanglement $E_\mathcal{N}(C_AC_B:P_AP_B)$ grows rapidly and saturates near its maximum value. This behaviour of collective bipartite entanglement is qualitatively similar to the product input state, and we can comment that this oscillation is the artifact of the joint coin-position unitary operator describing the DTQW model. \\
\indent Note that, unlike the pure product initial state in the coin-coin system, we can not apply the additivity property of the Logarithmic Negativity in the joint coin:position, we have found that the collective entanglement in the coin:position bipartition is superadditive ($E_\mathcal{N}(C_AC_B: P_AP_B) > E_\mathcal{N}(C_A:P_A) + E_\mathcal{N}(C_B:P_B)$) in nature. The observed superadditivity is due to the initial coin-coin entanglement in the system.

Fig.~\ref{fig:N51_T25_entangled}(b) depicts that the initial coin-coin entanglement is rapidly redistributed among all the subsystems during the evolution, among all the subsystem. Starting from its maximal value at $t=0$, the coin-coin entanglement, quantified by the bipartite entanglement measure, logarithmic negativity $E_\mathcal{N}(C_A:C_B)$, quickly disappears (as it is clear from Eq. \eqref{eq:Bpsi1} - \eqref{eq:Bpsi5}), at the very first action of the DTQW unitary $U_{AB}$,  and remains zero thereafter for all $t > 1$.  A small but finite position-position entanglement $E_\mathcal{N}(P_A:P_B)$, has also been generated during the time evolution, which was absent for product input state. These observations indicate that the quantum correlations initially localized within the coin subsystem are progressively transferred to the larger Hilbert space comprising both coin and position degrees of freedom. In cross-correlations the entanglement, ($E_\mathcal{N}(C_A:P_B)$ = $E_\mathcal{N}(C_B:P_A)$), remains zero throughout the evolution, with a very small non-zero value at $t = 3$ and $t=4$. Thus, although the initial Bell-state entanglement spreads throughout the system, the dominant contribution to the bipartite correlations originates from the coupling between each coin and its corresponding position degree of freedom. These results indicate that the entanglement initially present in the Bell state is redistributed across the four-partite system as the quantum walk evolves. However, the largest share of the entanglement remains concentrated between the coin and position degrees of freedom of each individual walker. This behavior can be attributed to the local coin-position coupling generated by the single-walker evolution operator.\\
\indent The most striking feature of the dynamics is revealed when we compute the genuine multipartite entanglement content in the entire system, 
by calculating the GGM of $\ket{\Psi(t)}_{AB}$ (see  Fig.~\ref{fig:N51_T25_entangled}(c)).
The GGM is initially zero because the state $\ket{\Psi(t=0)}_{AB}$ is separable across the collective coin-position bipartition. Following the first step of the quantum walk, it increases sharply to its maximum value and remains there during the initial stages of the evolution. At later times, the GGM exhibits oscillatory dynamics and repeatedly returns to its maximum value. Since the theoretical upper bound of the GGM is $1/2$ (see Appendix~\ref{Max_GGM}), these revivals signify that the coin-position interactions continuously redistribute the initial bipartite entanglement of the coins into genuine multipartite correlations involving all four subsystems, $C_A$, $P_A$, $C_B$, and $P_B$.\\
\indent An even more striking feature emerges from the lower envelope of the oscillations, which can be obtained by connecting the successive local minima of $GGM(\ket{\Psi(t)}_{AB})$ of Fig.~\ref{fig:N51_T25_entangled}(c). The local minima increase monotonically with time and asymptotically approach the theoretical upper bound, $GGM=1/2$. Thus, while the GGM repeatedly reaches its maximum value throughout the evolution, even its minimum value progressively converges to the same limit. Consequently, the amplitude of the oscillations steadily diminishes, and the dynamics become increasingly confined to a narrow region around the maximal GGM. This behaviour demonstrates that the initial bipartite entanglement is progressively redistributed into persistent genuine multipartite correlations among the four subsystems, $C_A$, $P_A$, $C_B$, and $P_B$. At sufficiently long times, the system evolves in a regime where the genuine multipartite entanglement remains close to its theoretically allowed maximum.\\
\indent Since the present regime satisfies $N>2T$, the walkers never encounter the lattice boundaries and therefore evolve as if they were propagating on an effectively infinite line. Consequently, the observed approach of the GGM towards its maximal value is an intrinsic feature of the unitary quantum-walk dynamics and is not influenced by finite-size effects or boundary-induced recurrences. These results demonstrate that an initially entangled coin state, together with the conditional shift dynamics of the quantum walk, provides an efficient mechanism for generating and sustaining near-maximal genuine multipartite entanglement.

\begin{figure}[b]
    \centering
    \subsubsection*{For N=51}
    \includegraphics[width=1\linewidth]{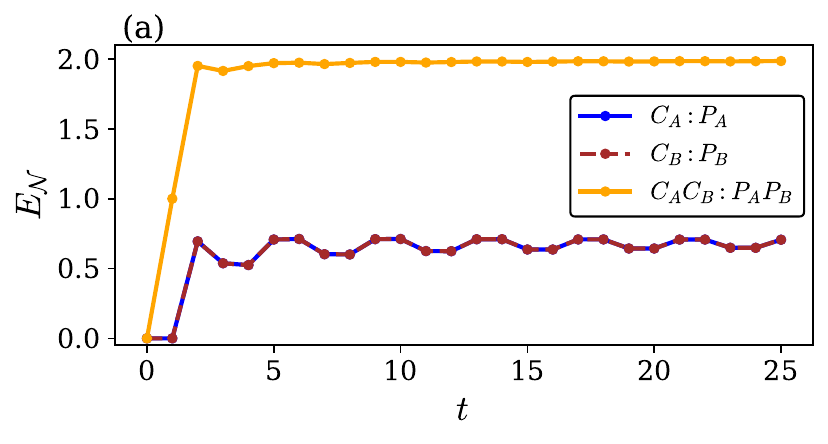}
    \includegraphics[width=1\linewidth]{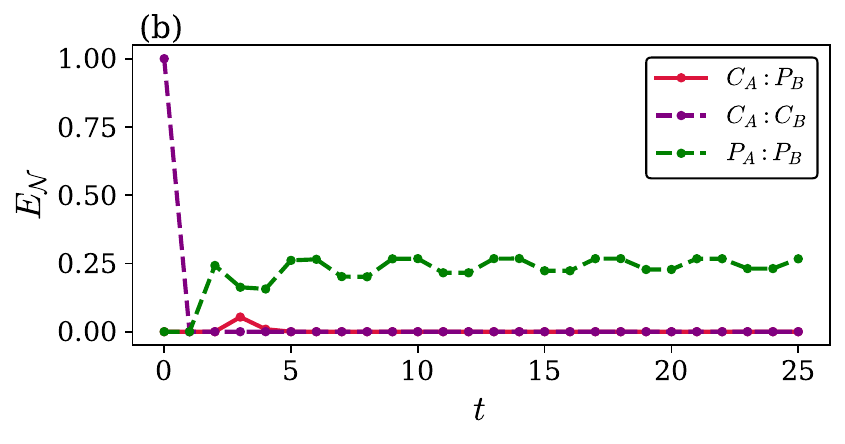}
    \caption{(Color online.)
    Time evolution of bipartite entanglement (in ebit unit) in various bipartitions has been plotted as a function of number of steps $t$ (dimensionless), in the two-walker DTQW, when the initial coin: coin quantum state is $\ket{\psi_{4}}_{C_AC_B}=\frac{1}{\sqrt{2}} \left( \ket{01}+\ket{10} \right)_{C_AC_B}$ The entanglement has been calculated with the help of logarithmic negativity in the bipartitions (a) $C_A:P_A$, $C_B:P_B$ and $C_AC_B:P_AP_B$, and (b) $C_A:P_B$ , $C_A:C_B$ and $P_A: P_B$. Here, we consider the open-boundary regime with $N=51$ lattice points and a maximum allowed step count of $T=25$.}
    \label{fig:N51_T25_entangled_psi_plus}
\end{figure}

For the initial coin state $\ket{\Phi^+}$, the entanglement dynamics are exactly identical to those depicted in Fig. \ref{fig:N51_T25_entangled}. In contrast, when the initial coin-coin state is chosen to be either $\ket{\Phi^-}$ or $\ket{\Psi^+}$, the bipartite entanglement dynamics show similar qualitative behavior but differs considerably
 from Fig.  \ref{fig:N51_T25_entangled}. The variation is shown in the following Fig. \ref{fig:N51_T25_entangled_psi_plus} for $\ket{\Psi^+}$.
Although the dynamics of the generalized geometric measure (GGM) remains unchanged, as will be discussed in the next section.

This behaviour can be understood from the symmetry properties of the Bell basis. The four Bell states naturally separate into two equivalence classes,
\begin{equation}
    \{\ket{\Psi^-},\ket{\Phi^+}\} \quad \{\ket{\Psi^+},\ket{\Phi^-}\},
\end{equation}
whose members are related by a local $\sigma_y$ operation acting on either coin state. Since $\sigma_y$ is the only Pauli operator that anti-commutes with the Hadamard coin operator, the resulting quantum walks are related by local unitary transformations combined with a spatial reflection, around the lattice site $x=\lfloor N/2\rfloor$. 
Focusing on the class $\{\ket{\Psi^-},\ket{\Phi^+}\}$, if  
$\ket{\Psi(t)}_{AB}$, and $\ket{\tilde{\Psi}(t)}_{AB}$, 
are the time evolved states, when the initial coin-coin state is $\ket{\Psi^-}$ and $\ket{\Phi^+}$ respectively, we have the following lemma.

\begin{lemma}
\label{lem:Bell_LU_equivalence}
For the two-walker DTQW considered here, the time-evolved states
corresponding to the initial coin states $\ket{\Psi^-}$ and
$\ket{\Phi^+}$ are related by a local unitary transformation, up to an
irrelevant global phase. In particular,
\begin{equation}
\ket{\widetilde{\Psi}(t)}_{AB}
=
(-1)^t
\left[
(\sigma_y)_{C_A}\otimes R_{P_A}
\right]
\ket{\Psi(t)}_{AB},
\label{eq:Bell_LU_relation}
\end{equation}
where 
$R_{P_A}$ is the reflection operator acting on the position
Hilbert space of walker $A$.
\end{lemma}

\begin{proof}
Note that the bell state $\ket{\phi^+}$, can be written as (neglecting the global phase)
\begin{equation}\label{eq:unitary-equivalent}
    \ket{\Phi^+} = (\sigma_y \otimes I)\ket{\Psi^-}  = (I \otimes \sigma_y) \ket{\Psi^-},
\end{equation} 
and hence the (tilde)  $\ket{\tilde{\Psi}(t)}_{AB}$, can be expressed as 
\begin{equation}\label{eq:PSITPHI}
   \ket{\tilde{\Psi}(t)}_{AB} =  (U_A)^t \otimes (U_B)^t \ket{\tilde{\Psi}(0)}_{AB},
\end{equation}
where 
$\ket{\tilde{\Psi}(0)}_{AB}
=
\ket{\phi^+}_{C_AC_B}
\otimes
\ket{\lfloor N/2\rfloor,\lfloor N/2\rfloor}_{P_AP_B}$, represents the total initial state.
By using Eq. \eqref{eq:unitary-equivalent}, one can find 
\begin{eqnarray}
     \ket{\tilde{\Psi}(t)}_{AB} &=&  (U_A)^t \otimes (U_B)^t \ket{\tilde{\Psi}(0)}_{AB} \\
     &=& (U_A)^t \left((\sigma_y)_{C_A} \otimes I_{P_A} \right)\otimes (U_B)^t \nonumber \\
     && \hspace{3em}\ket{\psi^-}_{C_AC_B} 
\otimes
\ket{\lfloor N/2\rfloor,\lfloor N/2\rfloor}_{P_AP_B},~~~
\end{eqnarray}
where the unitary opeartors acting only on the the walker $A$'s subsystem, can be expanded as 
\begin{eqnarray}\label{eq:Unitary-evolvesigma}
    (U_A)^t \left((\sigma_y)_{C_A} \otimes I_{P_A} \right) = \left( S_A (H_{C_A} \otimes I_{P_A})\right)^t \left((\sigma_y)_{C_A} \otimes I_{P_A} \right).~~
\end{eqnarray}
By using $H\sigma_y = -\sigma_y H$, and $ S_A \left((\sigma_y)_{C_A} \otimes I_{P_A} \right) = \left((\sigma_y)_{C_A} \otimes R_{P_A} \right)S_A\left(I_{C_A} \otimes R_{P_A} \right)$, where the hermitian $R_{P_A} = \sum_x \ket{N-x}\bra{x}_{P_A}$, denote the parity (reflection) operator  about the lattice site $x=\lfloor N/2\rfloor$, 
one can find that Eq. \eqref{eq:Unitary-evolvesigma}, transform as 
\begin{eqnarray}
    (U_A)^t \left((\sigma_y)_{C_A} \otimes I_{P_A} \right) &=& (-1)^t \left((\sigma_y)_{C_A} \otimes R_{P_A} \right) \times \nonumber \\ 
    && \hspace{-4em}\left( S_A (H_{C_A} \otimes I_{P_A})\right)^t \left(I_{C_A} \otimes R_{P_A} \right).
\end{eqnarray}
Now plug it into the expression of $\ket{\tilde{\Psi}(t)}_{AB}$, in Eq. \eqref{eq:PSITPHI}, and consider the fact that $R_{P_A} \ket{\lfloor N/2\rfloor} \approx \ket{\lfloor N/2\rfloor}$, we have obtained 
\begin{eqnarray}
     \ket{\tilde{\Psi}(t)}_{AB} &=&  
     (-1)^t \left((\sigma_y)_{C_A} \otimes R_{P_A} \right) \nonumber \\
    && \hspace{-2.8em}
\underbrace{(U_A)^t \otimes (U_B)^t \ket{\psi^-}_{C_AC_B} 
\otimes
\ket{\lfloor N/2\rfloor,\lfloor N/2\rfloor}_{P_AP_B}}_{\ket{{\Psi}(t)}_{AB}}, ~~~~ \\
&=&(-1)^t \left((\sigma_y)_{C_A} \otimes R_{P_A} \right) \ket{{\Psi}(t)}_{AB}.
\end{eqnarray}
This completes the proof.
\end{proof}
Hence, the two states differ only by a local unitary transformation acting on subsystems $C_A$ and $P_A$, together with the irrelevant global phase $(-1)^t$. One can equally put the $sigma_y$ in part of the coin state of walker $B$, and can show that the two time evolved state  differ by local unitaries acting on $C_B$ and $P_B$. 
Similar observation holds for the class $\{\ket{\Psi^+},\ket{\Phi^-}\}$, as stated in the following corollary. 
\begin{corollary}\label{coro:symmetric}
    In the two walker DTQW model, if the initial coin states are chosen from the set $\{\ket{\Psi^+},\ket{\Phi^-}\}$, then the corresponding time evolved states are connected through local unitary operator upto a global phase, acting in the subsystems $C_x$ and $P_x$, with $x \in \{A,B\}$.
\end{corollary}
\begin{proof}
    The proof is directly follows from the derivation in Lemma \ref{lem:Bell_LU_equivalence}, and the fact that upto a  global phase
\begin{equation}\label{eq:unitary-equivalent2}
    \ket{\Phi^-} = (\sigma_y \otimes I)\ket{\Psi^+}  = (I \otimes \sigma_y) \ket{\Psi^+},
\end{equation} 
\end{proof}

Based on these results we have the following theorem.  
\begin{theorem}
\label{thm:Bell_entanglement_equivalence}
The time evolution of the bipartite entanglement, under the discrete-time quantum walk, in all possible non-trivial  bipartition is identical for the set $\{\ket{\psi^-}, \ket{\phi^+}\}$ and $\{\ket{\psi^+}, \ket{\phi^-}\}$.
\end{theorem}

\begin{proof}
According to Lemma~\ref{lem:Bell_LU_equivalence}, the two
time-evolved states are related by
\begin{equation}
\ket{\widetilde{\Psi}(t)}_{AB}
=
(-1)^t
\left[
(\sigma_y)_{C_A}\otimes R_{P_A}
\right]
\ket{\Psi(t)}_{AB}.
\end{equation}

Hence, the two states differ only by a local unitary transformation acting on subsystems $C_A$ and $P_A$, while $(-1)^t$ is an irrelevant
global phase. As the bipartite entanglement does not change under local unitary transformation \cite{Horodecki2009}, therefore,  it  takes identical values, across any bipartition, for the initial coin coin states chosen from $\{\ket{\Psi^-},\ket{\Phi^+}\} $. \\
The same construction applies to the second set of Bell-state pair, as evident from Corollary \ref{coro:symmetric}.
Thus, the four Bell states naturally form two local-unitary
equivalence classes,
\begin{equation}
\left\{\ket{\Psi^-},\ket{\Phi^+}\right\},
\qquad
\left\{\ket{\Psi^+},\ket{\Phi^-}\right\},
\end{equation}

within which the corresponding time-evolved states have identical
bipartite entanglement across every bipartition.
\end{proof}

\subsection{Closed-boundary regime ($N<2T$)}
In this section, we forced both walkers to move within the confinement of the lattice sites, where the maximum number of lattice sites satisfies $N < 2T$, so that the walkers can encounter the lattice boundary multiple times within their maximum time steps, $T$. 

\subsubsection{Separable Coin-Coin State} In this subsection, we look into the entanglement dynamics when the initial coin state is separable, i.e., of the form $\ket{\psi}_{C_AC_B} = \ket{\phi}_{C_A}\otimes \ket{\chi}_{C_B}$,
in the closed-boundary regime.\\
\indent As discussed in the previous section, the non-trivial contributions are coming from the $C_AC_B:P_AP_B$ and $C_x:P_x, ~x \in \{A,B\}$ bipartitions. The plot of time evolution of entanglement in those bipartitions has been presented in Fig.~\ref{fig:N3_T50_sep}, as a function of discrete number of steps, for $N=3$, where the initial joint coin state is chosen to be $\ket{\psi_1}_{C_AC_B}=\ket{01}_{C_AC_B}$.\\
\begin{figure}[t]
    \centering
    \subsubsection*{For N=3}
    \label{subsection_N=3}
    \includegraphics[width=\linewidth]{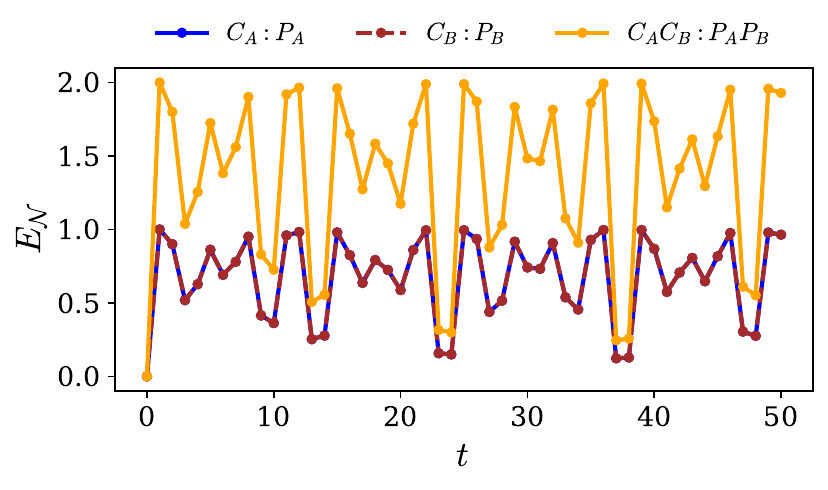}
    \caption{
    (Color online.)
    Time evolution of bipartite entanglement (in ebit unit) in various bipartitions has been plotted as a function of number of steps $t$ ( dimensionless), in the two-walker DTQW, when the initial coin: coin quantum state is a product state $\ket{\psi_{1}}_{C_AC_B}=\ket{01}_{C_AC_B}$. The entanglement has been calculated with the help of logarithmic negativity in the bipartitions $C_A: P_A$ (by the solid blue line) and $C_B: P_B$ (by the dashed brown line) $C_AC_B: P_AP_B$ (by the solid orange curve). Here, we consider the closed-boundary regime with $N=3$ lattice points and a maximum allowed step count of $T=50$.}
    \label{fig:N3_T50_sep}
\end{figure}
Similar to the open-boundary regime, we have found that for this particular initial state:  $E_\mathcal{N}(C_A:P_A) = E_\mathcal{N}(C_B:P_B)$, and the collective coin:position entanglement reduces to $E_\mathcal{N}(C_AC_B: P_AP_B)  = 2E_\mathcal{N}(C_A:P_A)$. The evolution of entanglement exhibits several noteworthy features that can be directly understood from the structure of the two-walker quantum walk.\\
The solid orange curve in Fig.~\ref{fig:N3_T50_sep} corresponds to the collective coin-position entanglement $E_\mathcal{N}(C_AC_B: P_AP_B)$, while the solid blue curve and dashed brown curve represent the local coin-position entanglements $E_\mathcal{N}(C_A: P_A)$ and $E_\mathcal{N}(C_B: P_B)$, respectively.
We observe that these quantities exhibit persistent oscillations throughout the evolution, which are quasi-periodic in nature. The oscillatory behavior originates from the repeated exchange of quantum correlations between the coin and position degrees of freedom through the conditional shift operation. Since the lattice consists of only three sites with periodic boundary conditions, the walkers repeatedly revisit the same lattice sites, leading to strong interference and recurrence effects that sustain the observed oscillations.

We have also investigated the dynamics for another initial coin product state, $\ket{0+}_{C_AC_B}$, where $\ket+=\frac{1}{\sqrt2}(\ket{0}+\ket1)$. The dynamics is mostly similar with a major difference of $E_\mathcal{N}(C_A:P_A)$ and $E_\mathcal{N}(C_B:P_B)$ being inequal as plotted in Fig. \ref{fig:N3_T50_sep_0+}.

\begin{figure}[b]
    \subsubsection*{For N=3}  
    \includegraphics[width=\linewidth]{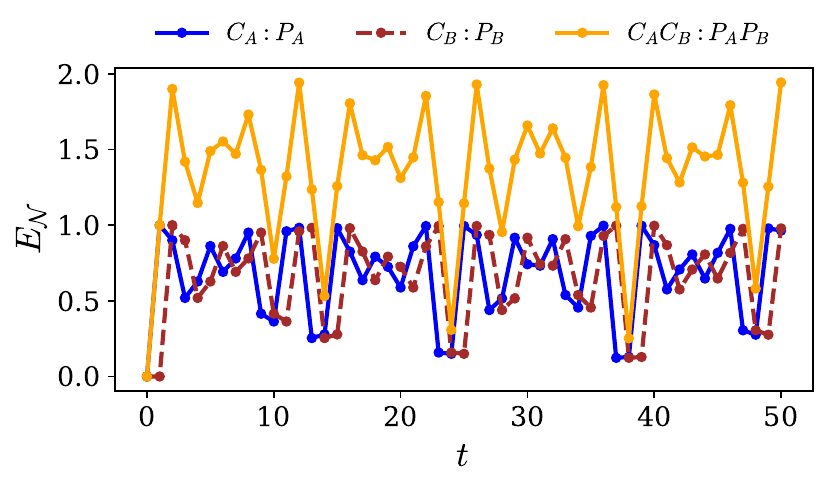}
    \caption{ (Color online.)
    Time evolution of bipartite entanglement (in ebit unit) in various bipartitions has been plotted as a function of number of steps $t$ ( dimensionless), in the two-walker DTQW, when the initial coin: coin quantum state is a product state $\ket{\psi_{2}}_{C_AC_B}=\ket{0+}_{C_AC_B}$. The entanglement has been calculated with the help of logarithmic negativity in the bipartitions $C_A: P_A$ (by the solid blue line) and $C_B: P_B$ (by the dashed brown line) $C_AC_B: P_AP_B$ (by the solid orange curve). Here, we consider the closed-boundary regime with $N=3$ lattice points and a maximum allowed step count of $T=50$. }
    \label{fig:N3_T50_sep_0+}
\end{figure}

The qualitative features observed for $N=4$, given in the Fig. \ref{fig:N4_T50_sep} are largely identical to those discussed previously for $N=3$. In particular, the coin-coin entanglement $E_\mathcal{N}(C_A: C_B)$, the position-position entanglement $E_\mathcal{N}(P_A: P_B)$, and the cross coin-position correlations remain identically zero throughout the evolution.  Similarly, the generalized geometric measure vanishes at all times, reflecting the absence of genuine multipartite entanglement.\\
\indent The principal difference between the $N=3$ and $N=4$ lattices lies in the appearance of an exact periodic structure in the entanglement dynamics for the latter case. As evident from Fig.~\ref{fig:N4_T50_sep}, both the collective coin-position entanglement $E_\mathcal{N}(C_AC_B: P_AP_B)$ and the local coin-position entanglement $E_\mathcal{N}(C_A: P_A)=E_\mathcal{N}(C_B :P_B)$ exhibit perfectly regular oscillations with complete revivals. Unlike the quasi-periodic oscillations observed for $N=3$, the entanglement dynamics for $N=4$ repeats exactly after a fixed number of time steps. The observed periodicity can be understood from the spectral properties of the walk operator. For a discrete-time quantum walk on a cycle, a periodic solution exists if there is an integer $\tau$ such that
\begin{equation}
    U_{A(B)}^\tau = I.
    \label{eq:periodicity_condition}
\end{equation}
\indent For the Hadamard walk on a cycle, Portugal \cite{2462630} has shown that the eigenphases satisfy the condition
\begin{equation}
    \sin\left(\frac{2\pi j_k}{\tau}\right)=\frac{1}{\sqrt{2}}\sin\left(\frac{2\pi k}{N}\right), \hspace{0.1in}\forall \hspace{0.1in}0\le k\le N-1,
    \label{eq:Portugal_condition}
\end{equation}
where $\tau$ denotes the period of the walk. Periodic solutions exist only when Eq.~(\ref{eq:Portugal_condition}) admits integer solutions for all momentum modes, $k$. In particular, for the Hadamard walk on a cycle, the smallest nontrivial periodic lattices are $N=2$ with period $\tau=2$, $N=4$ with period $\tau=8$, and $N=8$ with period $\tau=24$~\cite{2462630}.

\begin{figure}[t]
    \centering
    \subsubsection*{For N=4}
    \includegraphics[width=\linewidth]{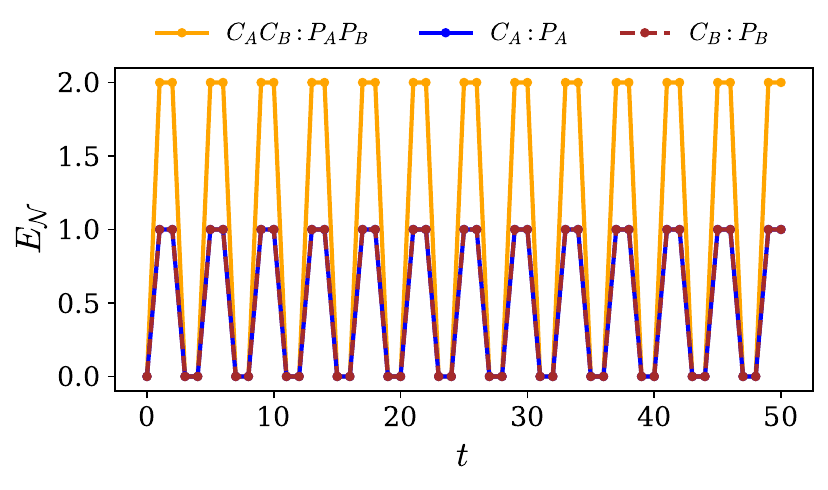}
    \caption{ (Color online.) Time evolution of bipartite entanglement (in ebit unit) in various bipartitions has been plotted as a function of number of steps $t$ ( dimensionless), in the two-walker DTQW, when the initial coin: coin quantum state is a product state $\ket{\psi_{1}}_{C_AC_B}=\ket{01}_{C_AC_B}$. The entanglement has been calculated with the help of logarithmic negativity in the bipartitions $C_A: P_A$ (by the solid blue line) and $C_B: P_B$ (by the dashed brown line) $C_AC_B: P_AP_B$ (by the solid orange curve). Here, we consider the closed-boundary regime with $N=4$ lattice points and a maximum allowed step count of $T=50$.}
    \label{fig:N4_T50_sep}
\end{figure}

\begin{figure}[H]
    \subsubsection*{For N=4}  
    \includegraphics[width=1\linewidth]{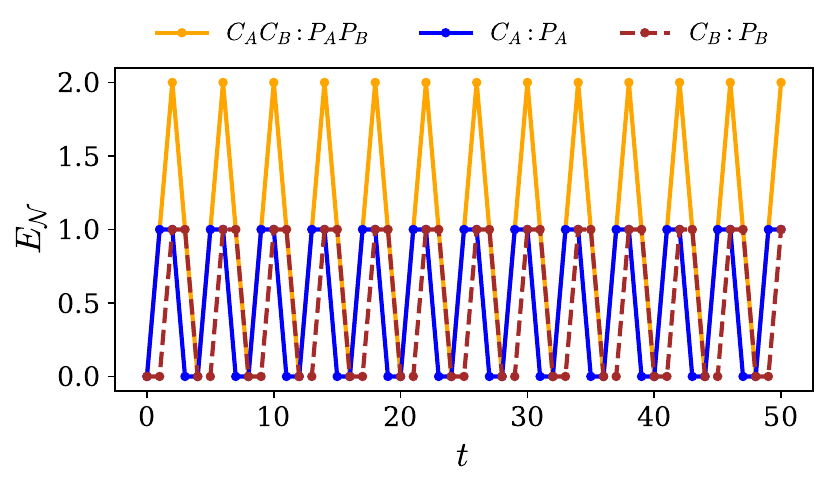}
    \caption{  (Color online.)
    Time evolution of bipartite entanglement (in ebit unit) in various bipartitions has been plotted as a function of number of steps $t$ ( dimensionless), in the two-walker DTQW, when the initial coin: coin quantum state is a product state $\ket{\psi_{2}}_{C_AC_B}=\ket{0+}_{C_AC_B}$. The entanglement has been calculated with the help of logarithmic negativity in the bipartitions $C_A: P_A$ (by the solid blue line) and $C_B: P_B$ (by the dashed brown line) $C_AC_B: P_AP_B$ (by the solid orange curve). Here, we consider the closed-boundary regime with $N=4$ lattice points and a maximum allowed step count of $T=50$.}
    \label{fig:N4_T50_sep_0+}
\end{figure}
For the other initial coin product state, $\ket{0+}_{C_AC_B}$, where $\ket+=\frac{1}{\sqrt2}(\ket{0}+\ket1)$, the dynamics is mostly similar with a major difference of $E_\mathcal{N}(C_A:P_A)$ and $E_\mathcal{N}(C_B:P_B)$ being inequal as plotted in Fig. \ref{fig:N4_T50_sep_0+}.

One important thing to notice is that the periodicity of the entanglement is equal to 4 while the periodicity of the operator is 8 for N=4. As shown analytically in Appendix~\ref{Analytical_N=4}, after every four time steps, the coin state is restored \cite{PhysRevA.95.052338} while the position state is translated by two lattice sites modulo four,
\begin{equation*}
    |0\rangle \leftrightarrow |2\rangle,\qquad|1\rangle \leftrightarrow |3\rangle.
\end{equation*}
Consequently, all entanglement measures exhibit exact revivals with a time period of 4 during the evolution.

Therefore, the exact revivals observed in Fig.~\ref{fig:N4_T50_sep} and Fig.~\ref{fig:N4_T50_sep_0+}  are a direct consequence of the fact that the four-site cycle constitutes a periodic quantum walk. The regular oscillations of the entanglement measures are thus manifestations of the recurrence properties of the underlying unitary evolution rather than the generation of additional quantum correlations. 
\label{subsection_N=4}
\subsubsection{Maximally Entangled Coin-Coin State} Now we have investigated the entanglement dynamics when the initial coin state is  $\ket{\psi_{3}}_{C_AC_B}
=\frac{1}{\sqrt{2}}
\left(
\ket{01}-\ket{10}
\right)_{C_AC_B}$, in the closed-boundary region.
\begin{figure*}[ht]
\subsubsection*{For $N=3$}
\centering
\begin{minipage}[ht]{0.49\linewidth}
    \centering
    \includegraphics[width=0.98\linewidth]{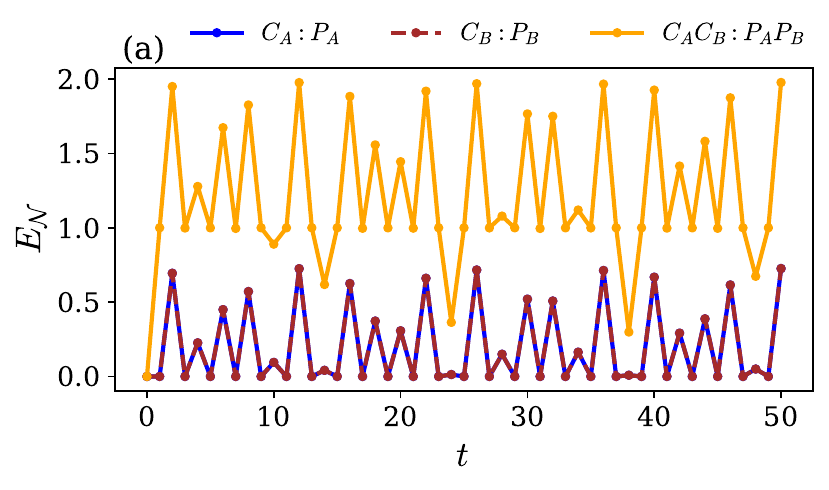}
    \includegraphics[width=0.98\linewidth]{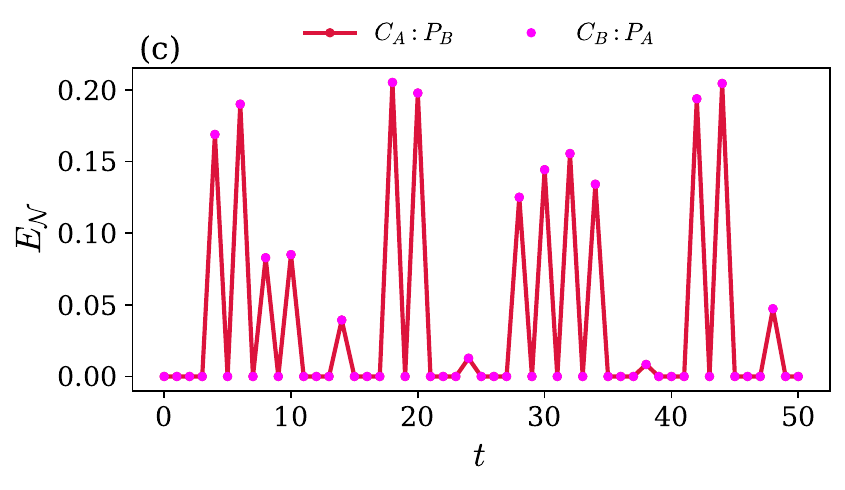}
\end{minipage}
\hfill
\begin{minipage}[ht]{0.49\linewidth}
    \centering
    \includegraphics[width=0.98\linewidth]{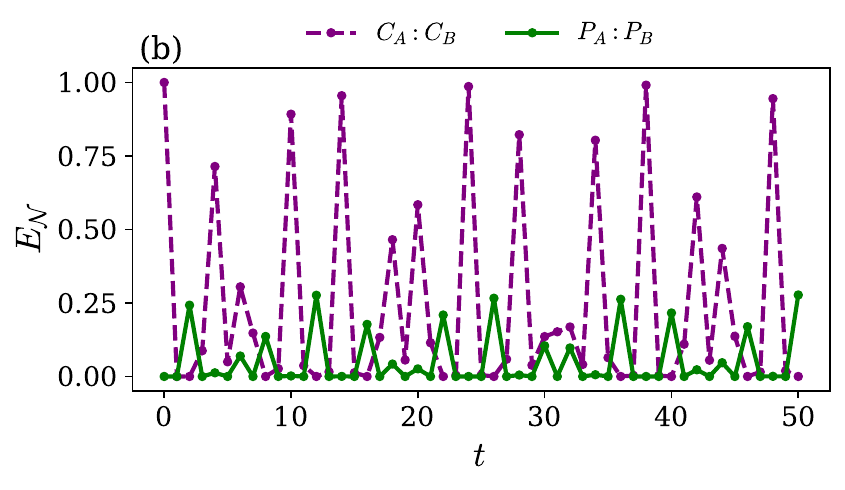}
    \vspace{0.1in}\\
    \includegraphics[width=0.98\linewidth]{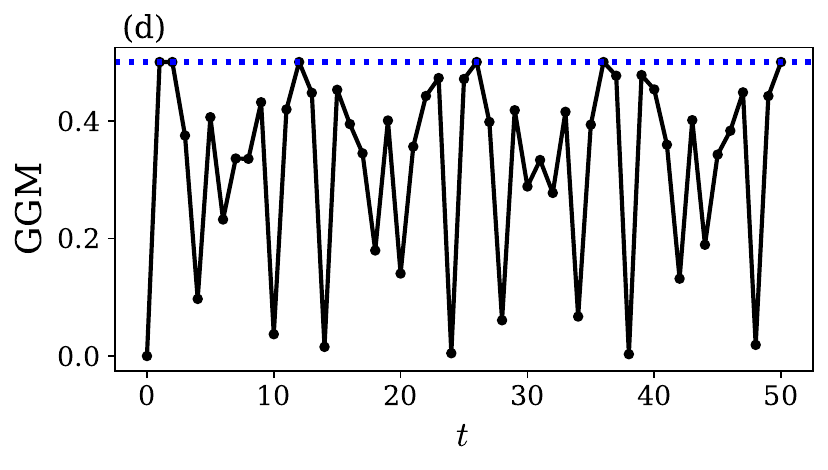}
\end{minipage}
\caption{(Color online.)
    Time evolution of bipartite entanglement (in ebit unit) in various bipartitions and GGM (dimensionless) has been plotted as a function of number of steps $t$ (dimensionless), in the two-walker DTQW, when the initial coin: coin quantum state is $\ket{\psi_{3}}_{C_AC_B}=\frac{1}{\sqrt{2}} \left( \ket{01}-\ket{10} \right)_{C_AC_B}$ The entanglement has been calculated with the help of logarithmic negativity in the bipartitions (a) $C_A:P_A$, $C_B:P_B$ and $C_AC_B:P_AP_B$, and (b) $C_A:C_B$, $P_A:P_B$ and (c) $C_A:P_B$ , $C_B:P_A$ and whereas in (d) the generalized geometric measure (GGM) of $\ket{\psi(t)}_{AB}$ has been plotted. Here, we consider the closed-boundary regime with $N=3$ lattice points and a maximum allowed step count of $T=50$.}
\label{fig:N3_T50_entangled}
\end{figure*}

\begin{figure*}[ht]
\subsubsection*{For $N=4$}
\centering
\begin{minipage}[ht]{0.49\linewidth}
    \centering
    \includegraphics[width=\linewidth]{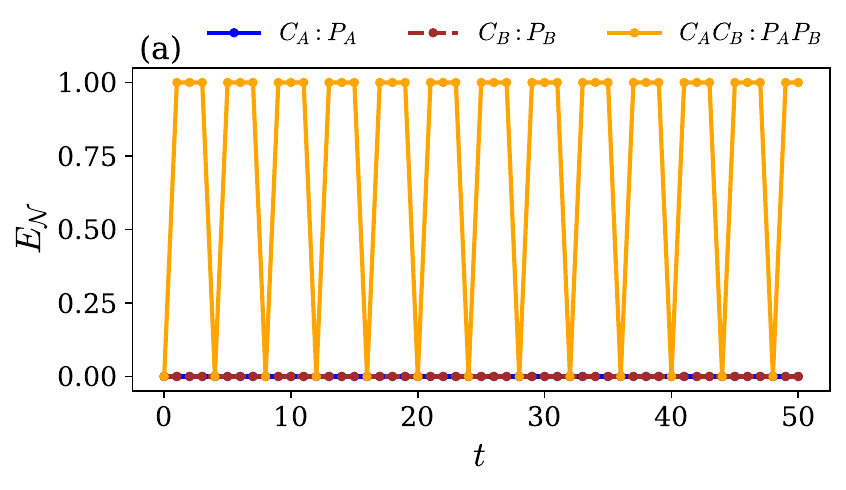}
    \includegraphics[width=\linewidth]{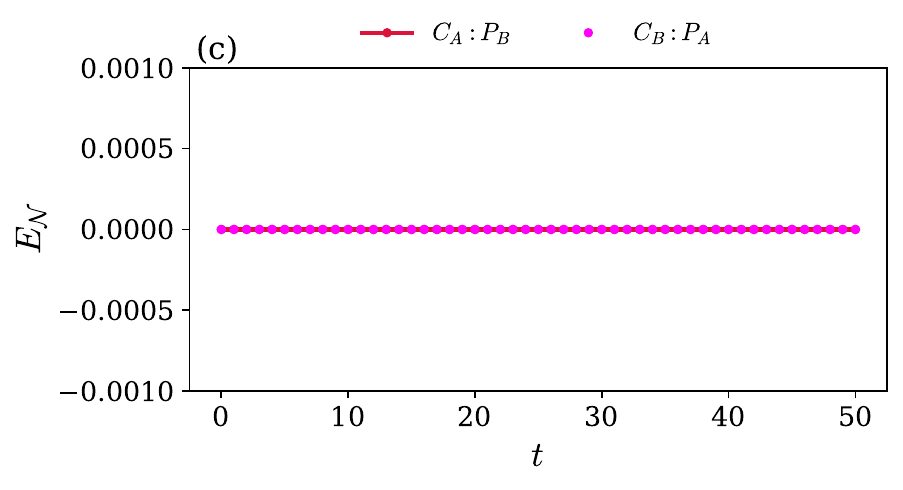}
\end{minipage}
\hfill
\begin{minipage}[ht]{0.49\linewidth}
    \centering
    \includegraphics[width=\linewidth]{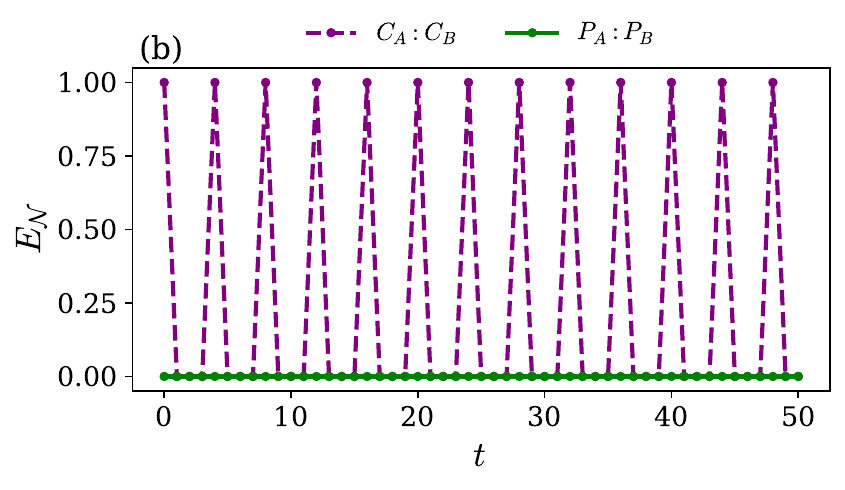}
    \includegraphics[width=\linewidth]{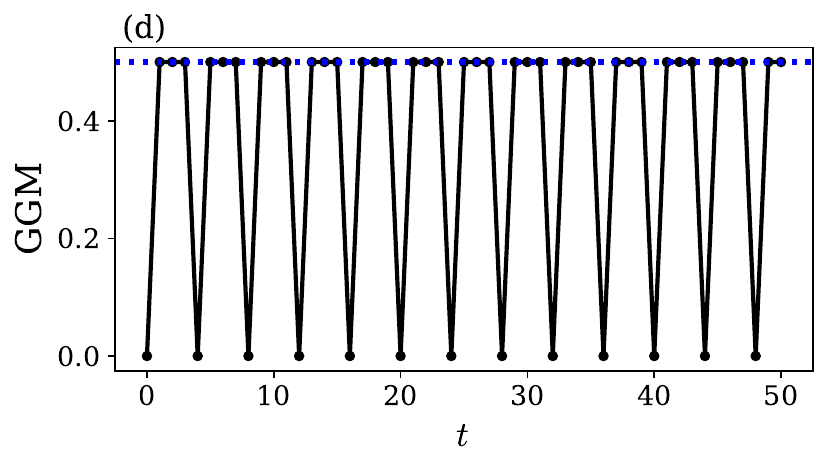}
\end{minipage}
\caption{(Color online)
    Time evolution of bipartite entanglement (in ebit unit) in various bipartitions and GGM (dimensionless) has been plotted as a function of number of steps $t$ (dimensionless), in the two-walker DTQW, when the initial coin: coin quantum state is $\ket{\psi_{3}}_{C_AC_B}=\frac{1}{\sqrt{2}} \left( \ket{01}-\ket{10} \right)_{C_AC_B}$ The entanglement has been calculated with the help of logarithmic negativity in the bipartitions (a) $C_A:P_A$, $C_B:P_B$ and $C_AC_B:P_AP_B$, and (b) $C_A:C_B$, $P_A:P_B$ and (c) $C_A:P_B$ , $C_B:P_A$ and whereas in (d) the generalized geometric measure (GGM) of $\ket{\psi(t)}_{AB}$ has been plotted. Here, we consider the closed-boundary regime with $N=4$ lattice points and a maximum allowed step count of $T=50$.}
\label{fig:N4_T50_entangled}
\end{figure*}
In contrast to the separable initial coin-coin case, the evolution now exhibits nontrivial quantum correlations across all sectors of the system. As shown in Fig.~\ref{fig:N3_T50_entangled}(a), the local coin-position entanglements $E_\mathcal{N}(C_A:P_A)$ and $E_\mathcal{N}(C_B:P_B)$ are generated immediately after the walk begins and exhibit identical oscillatory behavior due to the symmetry of the walk. Simultaneously, the collective coin-position entanglement $E_\mathcal{N}(C_AC_B:P_AP_B)$ develops rapidly and remains significantly larger than all other bipartite entanglement measures throughout the evolution. This behavior indicates that the initially localized coin-coin entanglement is continuously redistributed between the coin and position degrees of freedom by the conditional shift operation.

The dynamics of the coin-coin and position-position entanglement are shown in Fig.~\ref{fig:N3_T50_entangled}(b). The initial coin-coin entanglement is not preserved within the coin subsystem. Instead, $E_\mathcal{N}(C_A:C_B)$ undergoes pronounced oscillations and periodically decreases from its initial value. Furthermore, a small but nonzero position-position entanglement $E_N(P_A:P_B)$ is generated during the evolution, demonstrating that quantum correlations initially stored in the coin sector are partially transferred to the position sector.

The cross coin-position entanglement dynamics shown in Fig.~\ref{fig:N3_T50_entangled}(c) provides further insight into the redistribution process. 
The quantities $E_\mathcal{N}(C_A:P_B)$ and $E_\mathcal{N}(C_B:P_A)$ exhibit identical oscillatory behavior due to the symmetry of both the initial Bell state and the evolution operator under the exchange of the two walkers. Their magnitudes remain substantially smaller than the corresponding local coin-position entanglements at all steps as shown in Fig.~\ref{fig:N3_T50_entangled}(a). The nonzero values of these cross partitions indicate that the initial coin-coin entanglement allows quantum correlations to spread across different subsystems of the composite Hilbert space, although these correlations remain weaker than the local coin-position entanglement generated within each walker.

A particularly striking feature of the dynamics is the behavior of the generalized geometric measure shown in Fig.~\ref{fig:N3_T50_entangled}(d). Unlike the separable initial state case, where the GGM remains identically zero, the present evolution generates substantial genuine multipartite entanglement. The GGM rapidly increases from zero and repeatedly approaches the theoretical upper bound of $1/2$ derived in the Appendix~\ref{Max_GGM}. This demonstrates that the initially bipartite Bell-state entanglement is converted into genuine four-partite entanglement involving the subsystems $C_A$, $P_A$, $C_B$, and $P_B$. The repeated oscillations of the GGM reflect the continuous redistribution of quantum correlations among different bipartitions of the system.\\
\indent It is important to note that the evolution operator still factorizes as Eq.~(\ref{eq:U_AB}). Therefore, no new inter-walker interaction is introduced during the dynamics. The emergence of genuine multipartite entanglement is entirely due to the presence of initial coin-coin entanglement. Local unitary evolution converts this pre-existing nonlocal resource into multipartite correlations by coupling each coin degree of freedom to its corresponding position space. Consequently, the quantum walk acts as an efficient mechanism for transforming bipartite Bell-state entanglement into multipartite entanglement distributed across the entire four-partite system.\\
\indent Overall, the results demonstrate that the nature of the initial coin state plays a crucial role in determining the entanglement structure of the evolved state. While a separable initial state generates only local coin-position entanglement, an initially entangled Bell state leads to the formation of substantial genuine multipartite entanglement together with a rich redistribution of quantum correlations among the various bipartitions of the system.

The entanglement dynamics for the four-site cycle shown in Fig.~\ref{fig:N4_T50_entangled} exhibits a remarkably regular periodic structure, in sharp contrast to the quasi-periodic oscillations observed for $N=3$. The periodicity originates from the recurrence properties of the underlying quantum walk operator on the four-site cycle and is reflected in all entanglement measures as discussed in subsection~\ref{subsection_N=4}.\\
\indent As shown in Fig.~\ref{fig:N4_T50_entangled}(a), the collective coin-position entanglement $E_\mathcal{N}(C_AC_B:P_AP_B)$ oscillates periodically between zero and one. In contrast, the local coin-position entanglements $E_\mathcal{N}(C_A:P_A)$ and $E_\mathcal{N}(C_B:P_B)$ remain identically zero throughout the evolution. Thus, unlike the $N=3$ case, no local coin-position entanglement is generated. Instead, the dynamics consists entirely of a periodic redistribution of the initially present Bell-state entanglement between the coin subsystem and the global coin-position partition.\\
\indent The complementary behavior of the coin-coin and position-position entanglement is shown in Fig.~\ref{fig:N4_T50_entangled}(b). The coin-coin entanglement $E_\mathcal{N}(C_A:C_B)$ oscillates periodically between zero and one. Whenever the coin-coin entanglement attains its maximum value, the collective coin-position entanglement shown in Fig.~\ref{fig:N4_T50_entangled}(a) vanishes, and vice versa. This complementary behavior demonstrates a coherent transfer of quantum correlations between the coin subsystem and the collective coin-position bipartition. Throughout the evolution, the position-position entanglement $E_\mathcal{N}(P_A:P_B)$ remains identically zero, indicating that the quantum correlations never become localized exclusively within the position degrees of freedom.

An important feature of the dynamics is shown in Fig.~\ref{fig:N4_T50_entangled}(c), where the cross coin-position entanglement measures $E_\mathcal{N}(C_A:P_B)$ and $E_\mathcal{N}(C_B:P_A)$ also vanish identically for all times. Consequently, neither local nor nonlocal coin-position bipartite entanglement is generated during the evolution. In particular, we observe that, for $t = 1,2,$ and $3 \pmod{4}$, all pairwise bipartite entanglements between any two subsystems vanish. The only nonzero logarithmic negativity is that across the collective bipartition $(C_A C_B):(P_A P_B)$, while the GGM simultaneously attains its maximum value of $0.5$. This indicates that the generated entanglement is entirely genuine multipartite in nature and exhibits the characteristic entanglement structure of a four-qubit GHZ state~\cite{greenberger2007goingbellstheorem}, whose two-party reduced states are separable despite possessing maximal genuine multipartite entanglement. This behavior persists periodically and is interrupted only at $t = 4n$, where $n \in \mathbb{N}$. At these time steps, the system returns to a product state between the collective coin and position degrees of freedom, the GGM vanishes, and the initial coin--coin entanglement is completely restored. As evident from Eqs.~\eqref{eq:Bpsi1}--\eqref{eq:Bpsi3}, although each position subsystem belongs to an $N$-dimensional Hilbert space, only two position basis states are populated for $t \neq 4n$. Consequently, the dynamics is effectively confined to a four-qubit Hilbert space, giving rise to a GHZ-like entanglement structure.

As shown in Fig.~\ref{fig:N4_T50_entangled}(d), the GGM oscillates periodically between zero and its theoretical upper bound of $1/2$, derived in Appendix~\ref{Max_GGM}. Consequently, the system repeatedly evolves into states possessing the maximum possible genuine multipartite entanglement permitted by the Hilbert-space structure. This observation is fully consistent with the analytical expressions in Eqs.~\eqref{eq:Bpsi1}--\eqref{eq:Bpsi3}, which show that, for $t=1,2,$ and $3 \pmod{4}$, the dynamics is confined to an effective four-qubit subspace. The resulting states exhibit the same entanglement structure as a four-qubit GHZ state, characterized by maximal genuine multipartite entanglement and vanishing pairwise bipartite entanglement.

The attainment of GGM $=1/2$ demonstrates an efficient conversion of the initial Bell-state entanglement into genuine four-partite entanglement involving the subsystems $C_A$, $P_A$, $C_B$, and $P_B$.

These results highlight a qualitative difference between the $N=3$ and $N=4$ lattices. While the three-site cycle distributes the initial Bell-state entanglement among several bipartitions, generating both local coin-position and multipartite correlations, the four-site cycle exhibits a much cleaner dynamics characterized by a periodic exchange between coin-coin entanglement and genuine multipartite entanglement. The system therefore alternates between states dominated by bipartite Bell-type correlations and states exhibiting maximal multipartite entanglement.\\
\indent Overall, the $N=4$ lattice provides an example of a perfectly periodic entanglement-transfer mechanism, where the initially localized Bell-state entanglement is periodically transformed into maximum genuine multipartite entanglement and subsequently restored, without generating any position-position, local coin-position, or cross coin-position entanglement.

\section{Robustness of the Genuine Multipartite Entanglement}
In section \ref{GGM_gen}, we demonstrated that the open-boundary regime provides an efficient mechanism for generating genuine multipartite entanglement, with the generalized geometric measure (GGM) rapidly approaching its maximum attainable value. An important question, however, is whether this behavior is a special feature of a particular choice of initial conditions and quantum-walk parameters, or whether it is a generic property of the dynamics.\\
In this section, we systematically investigate the robustness of the generated genuine multipartite entanglement in the open-boundary regime. Specifically, we examine the sensitivity of the GGM to variations in both the initial joint coin state and the local coin operator. We first analyze the dependence of the multipartite entanglement on the choice of the four Bell states used as the initial coin state. Subsequently, we investigate the robustness of the GGM against continuous deviations from the Bell state and the Hadamard coin by considering one-parameter families of the initial coin state and the local coin operator. This analysis enables us to identify the essential ingredients responsible for the robust generation of near-maximal genuine multipartite entanglement in two-walker discrete-time quantum walks.


\subsection{Dependence on the Bell states}
To investigate the robustness of the maximum genuine multipartite entanglement generated during the quantum walk in the open-boundary region, we study the time evolution of the generalized geometric measure (GGM) for all four Bell states used as the initial coin state. Lemma~\ref{lem:Bell_LU_equivalence} establishes that the time-evolved
states within each Bell-state equivalence class are related by local
unitary transformations and therefore possess identical Schmidt
spectra across every bipartition. Since the GGM is invariant under
local unitary transformations, the GGM dynamics is identical within
each equivalence class. Remarkably, as shown in
Fig.~\ref{fig:bell_ggm}, the two equivalence classes also exhibit
identical GGM dynamics, resulting in the same GGM variation for all
four Bell states.

\begin{figure}[ht]
    \centering
    \textbf{(a)}\hfill \textbf{(b)}\\[1mm]
    \includegraphics[width=0.49\columnwidth]{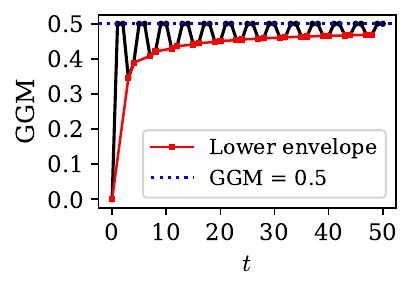}
    \hfill
    \includegraphics[width=0.49\columnwidth]{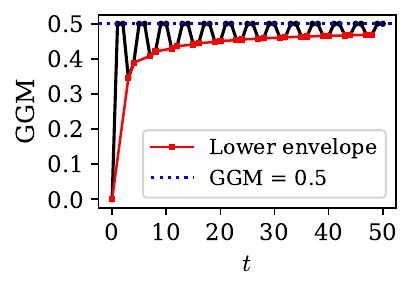}
    \textbf{(c)}\hfill \textbf{(d)}\\[1mm]
    \includegraphics[width=0.49\columnwidth]{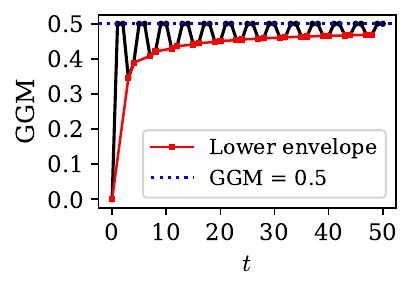}
    \hfill
    \includegraphics[width=0.49\columnwidth]{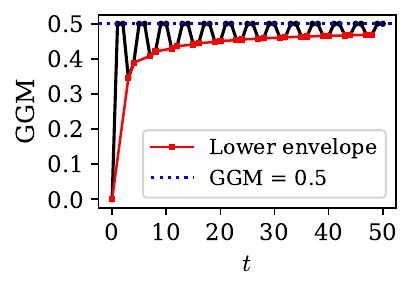}
    \caption{ (Color online)
    Time evolution of the generalized geometric measure (GGM) for the four Bell states used as the initial coin state:
    (a) $\ket{\Phi^{+}}=(\ket{00}+\ket{11})/\sqrt{2}$,
    (b) $\ket{\Phi^{-}}=(\ket{00}-\ket{11})/\sqrt{2}$,
    (c) $\ket{\Psi^{+}}=(\ket{01}+\ket{10})/\sqrt{2}$, and
    (d) $\ket{\Psi^{-}}=(\ket{01}-\ket{10})/\sqrt{2}$.
    The lattice size is fixed at $N=101$ ($N>2T$), corresponding to the open-boundary regime.}
    \label{fig:bell_ggm}
\end{figure}

Fig.~\ref{fig:bell_ggm} shows the variation of the GGM as a function of the number of time steps for the Bell states $\ket{\Phi^{+}}$, $\ket{\Phi^{-}}$, $\ket{\Psi^{+}}$, and $\ket{\Psi^{-}}$. We observe that all four Bell states exhibit identical GGM dynamics in the open-boundary regime. This demonstrates that the generation of genuine multipartite entanglement is insensitive to the particular choice of maximally entangled initial coin state. In particular, although the detailed bipartite entanglement dynamics can differ between the two Bell-state equivalence classes, the GGM remains unchanged. Thus, the robustness of the generated genuine multipartite entanglement is not tied to a specific Bell-state preparation, but instead reflects a more general feature of the two-walker quantum-walk dynamics.
\subsection{Robustness against Deviations from the Initial Bell State and the Coin Operator}

Having established that the generated genuine multipartite entanglement is identical for all four Bell states, we now investigate its robustness against perturbations in both the initial coin state and the local coin operator. To this end, we consider the one-parameter family of Bell-like states
given in Eq.~(\ref{eq:psi_eps}) and the one-parameter family of coin operators defined in Eq.~(\ref{eq:coin_eps}), parameterized by $\epsilon_1$ and $\epsilon_2$, respectively. Here, $\epsilon_1$ continuously interpolates between the maximally entangled Bell state and separable product states, whereas $\epsilon_2$ continuously deforms the Hadamard coin into other members of the one-parameter coin family.

To systematically examine the robustness of the generated genuine multipartite entanglement, we consider the following three cases:

\begin{enumerate}
    \item[(i)] $\epsilon_1\neq0$ and $\epsilon_2=0$, where only the initial coin state is varied while the Hadamard coin is kept unchanged.

    \item[(ii)] $\epsilon_1=0$ and $\epsilon_2\neq0$, where the initial Bell state is fixed and only the local coin operator is varied.

    \item[(iii)] $\epsilon_1\neq0$ and $\epsilon_2\neq0$, where both the initial coin state and the coin operator are simultaneously varied.
\end{enumerate}

These three cases allow us to distinguish the individual and combined effects of the initial-state entanglement and the quantum-walk dynamics on the generation of genuine multipartite entanglement.

\subsubsection*{Case I: Variation of the Initial Coin State ($\epsilon_1\neq0,\;\epsilon_2=0$)}

We first investigate the GGM for the following one-parameter family of Bell-like states

\begin{equation}
|\psi(\epsilon_1)\rangle_{C_AC_B}
=
\Big[
\sqrt{\frac{1+\epsilon_1}{2}}\,|01\rangle
-
\sqrt{\frac{1-\epsilon_1}{2}}\,|10\rangle
\Big]_{C_AC_B},
\label{eq:psi_eps}
\end{equation}

where $\epsilon_1\in[-1,1]$. The end points $\epsilon_1=\pm1$, correspond to the separable product states $|01\rangle$ and $|10\rangle$, whereas the centre $\epsilon_1=0$, retrieves the Bell-singlet state.
\begin{figure}[b]
    \centering
    \includegraphics[width=1.0\linewidth]{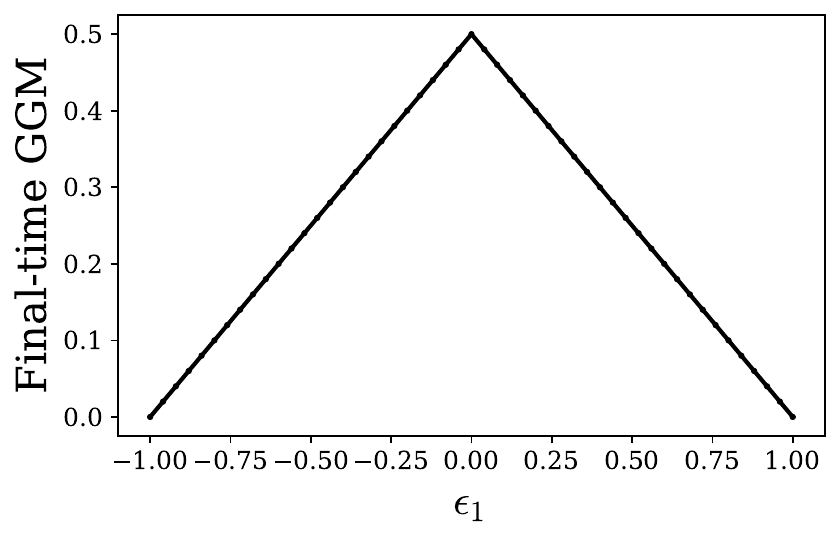}
    \caption{Final-time (at $t=50$) GGM as a function of the initial-state parameter $\epsilon_1$ for a two-walker DTQW in the open-boundary regime ($N=101$, $T=50$). The initial coin state is chosen from the family of Bell-like states given in Eq.~(\ref{eq:psi_eps}). The maximum GGM is obtained for the Bell-singlet state ($\epsilon_1=0$), while the GGM vanishes for the separable states ($\epsilon_1=\pm1$).}
    \label{fig:eps1_variation}
\end{figure}
The final-time GGM has been plotted in Fig.~\ref{fig:eps1_variation}, as a function of $\epsilon_1$. The generated multipartite entanglement exhibits a symmetric dependence about $\epsilon_1=0$, attaining its maximum value,
$GGM_{\rm max}=1/2$ (see Appendix~\ref{Max_GGM}),
for the Bell-singlet state and decreasing monotonically as the initial state approaches a product state. In fact, we have found that the numerical results satisfy

\begin{equation}
GGM_{\rm final}
=
0.5 (1-|\epsilon_1|).
\end{equation}

This demonstrates that the amount of genuine multipartite entanglement generated by the quantum walk is determined primarily by the entanglement initially present in the coin subsystem. The Bell-singlet therefore serves as the optimal initial resource for producing maximal multipartite entanglement.

\subsubsection*{Case II: Variation of the Coin Operator ($\epsilon_1=0,\;\epsilon_2\neq0$)}
\label{subsection5c}
We next examine the robustness of multipartite entanglement against changes in the local coin operation while fixing the initial state to the Bell-singlet. The local coin operator is chosen from the one-parameter family \cite{PhysRevA.77.032326}

\begin{equation}
O_{C_A}(\epsilon_2)
=
\begin{pmatrix}
\cos\!\left(\frac{\pi}{4}+\epsilon_2\right) &
\sin\!\left(\frac{\pi}{4}+\epsilon_2\right) \\
\sin\!\left(\frac{\pi}{4}+\epsilon_2\right) &
-\cos\!\left(\frac{\pi}{4}+\epsilon_2\right)
\end{pmatrix}=O_{C_B}(\epsilon_2)
\label{eq:coin_eps},
\end{equation}

where $\epsilon_2=0$ corresponds to the Hadamard coin, while $\epsilon_2=\pm\pi/4$ recover the Pauli-$X$ and Pauli-$Z$ coins.
\begin{figure}[b]
    \centering
    \includegraphics[width=1.0\linewidth]{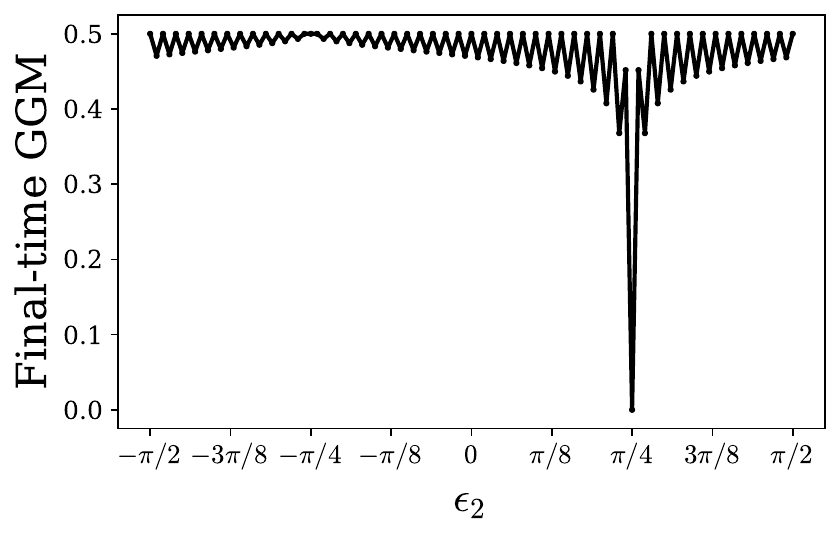}
    \caption{Final-time (at $t=50$) GGM as a function of the coin-operator parameter $\epsilon_2$ for a two-walker DTQW in the open-boundary regime ($N=101$, $T=50$). The walkers are initialized in the Bell-singlet state and evolved using the coin operator defined in Eq.~(\ref{eq:coin_eps}). The GGM remains close to its maximum value over most of the parameter range, with a pronounced dip near $\epsilon_2=\pi/4$.}
    \label{fig:eps2_variation}
\end{figure}
The corresponding final-time GGM is shown in Fig.~\ref{fig:eps2_variation}. We observe that the multipartite entanglement remains remarkably insensitive to the choice of coin operator over almost the entire parameter range. The GGM stays close to its theoretical maximum value of $1/2$, except within a narrow neighbourhood of $\epsilon_2=\pi/4$, where the Pauli-$X$ coin produces a pronounced suppression of multipartite entanglement (see Appendix~\ref{app:Xcoin}). This dip is a stroboscopic effect arising from evaluating the system at an even time step ($T=50$), where the Pauli-X coin forces the GGM to vanish while strictly preserving maximal GGM ($0.5$) at odd steps. Away from this isolated region, the GGM rapidly recovers to values very close to $1/2$. These results indicate that the generation of genuine multipartite entanglement is highly robust against imperfections in the local coin operation.

\subsubsection*{Case III: Simultaneous Variation of the Initial Coin State and Coin Operator ($\epsilon_1\neq0,\;\epsilon_2\neq0$)}

Finally, we investigate the combined influence of the initial coin state and the coin operator on the generated multipartite entanglement. Figure~\ref{fig:eps1_eps2} presents the final-time GGM over the two-dimensional parameter space $(\epsilon_1,\epsilon_2)$.
\begin{figure}[b]
\centering
\includegraphics[width=1\linewidth]{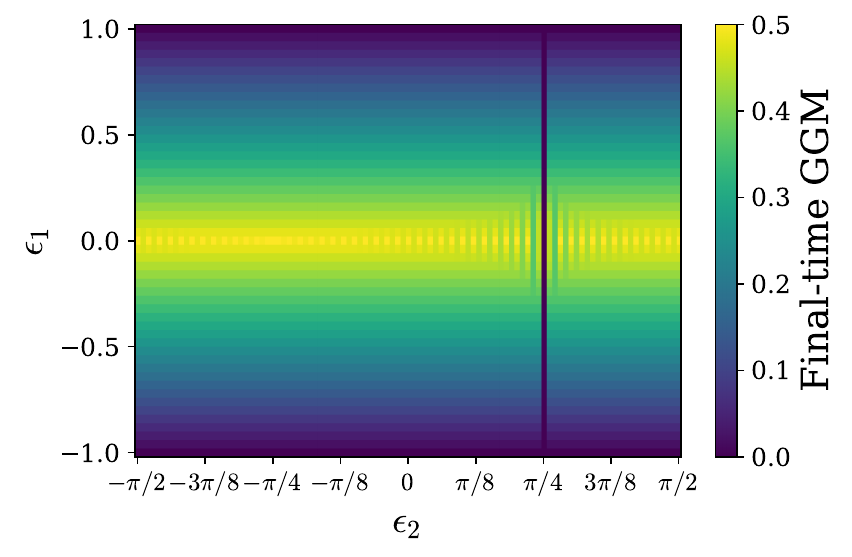}
\caption{ (Color online) Final-time (at $t=50$) GGM as a function of the initial-state parameter $\epsilon_1$ and the coin-operator parameter $\epsilon_2$ for a two-walker DTQW with $N=101$ and $T=50$. The initial coin state is given by Eq.~(\ref{eq:psi_eps}), while the evolution is generated using the coin operator of Eq.~(\ref{eq:coin_eps}).}
\label{fig:eps1_eps2}
\end{figure}

The two-dimensional landscape unifies the observations of the previous two cases. The vertical slice at $\epsilon_2=0$ reproduces the dependence on the initial coin state shown in Fig.~\ref{fig:eps1_variation}, whereas the horizontal slice at $\epsilon_1=0$ recovers the coin-operator dependence of Fig.~\ref{fig:eps2_variation}. Thus, the two one-parameter analyses correspond to special cross-sections of the complete parameter space.\\
\indent A prominent feature of Fig.~\ref{fig:eps1_eps2} is that the GGM varies predominantly along the $\epsilon_1$ direction. For nearly all values of $\epsilon_2$, the multipartite entanglement follows the same symmetric dependence on $\epsilon_1$, reaching its maximum at $\epsilon_1=0$ and decreasing continuously towards zero as $|\epsilon_1|$ approaches unity. In contrast, variations along the $\epsilon_2$ direction produce only minor changes, except for the narrow vertical strip around $\epsilon_2=\pi/4$, where the Pauli-$X$ coin suppresses the multipartite entanglement irrespective of the initial state.\\
\indent The combined parameter-space analysis therefore clearly separates the roles of the two control parameters. The initial-state parameter $\epsilon_1$ determines the attainable amount of genuine multipartite entanglement, whereas the coin-operator parameter $\epsilon_2$ only weakly influences its generation except in the vicinity of the Pauli-$X$ coin.

We can conclude that the maximal multipartite entanglement generation is not a fine-tuned feature of a specific Bell state or a particular choice of coin operator. It is a robust and generic feature of the open-boundary quantum walk, requiring only nonzero initial coin entanglement and remaining largely insensitive to the choice of the coin operator.

\section{Discussion and Conclusion}

We have investigated the generation and evolution of quantum correlations in a two-walker discrete-time quantum walk by treating the two coin and two position degrees of freedom as a four-partite quantum system. Employing logarithmic negativity and the generalized geometric measure (GGM), we have systematically characterized both bipartite and genuine multipartite entanglement generated during the quantum walk. Unlike previous studies that primarily focused on selected bipartitions or coin-position entanglement, the present work provides a unified picture of the redistribution of quantum correlations across all physically relevant subsystems.\\
\indent A central result of this work is the distinct role played by the lattice topology. In the open-boundary regime ($N>2T$), the walkers never encounter the lattice boundaries and the entanglement dynamics is governed solely by the coherent spreading of the wave packets. In this regime, the bipartite entanglement exhibits a well-defined and monotonic redistribution among different constituent subsystems, while the genuine multipartite entanglement approaches its theoretical upper bound, which is  $1/2$. In contrast, the closed-boundary regime ($N<2T$) introduces repeated interference arising from the periodic boundary conditions. The resulting multiple encounters of the wave packets lead to qualitatively different bipartite-entanglement dynamics, characterized by oscillatory behavior and repeated redistribution of quantum correlations among the various subsystem partitions. These observations demonstrate that the lattice topology provides an effective mechanism for controlling the flow of quantum correlations in discrete-time quantum walks.\\
\indent The present analysis also reveals a clear distinction between bipartite and multipartite entanglement. While the logarithmic negativities corresponding to different bipartitions exhibit strong sensitivity to the boundary conditions and undergo substantial temporal redistribution, the genuine multipartite entanglement displays remarkable robustness. In the open-boundary regime, the GGM rapidly approaches its maximum value and remains insensitive to the particular Bell state chosen as the initial coin state. Furthermore, by continuously varying both the initial coin state and the local coin operator, we have demonstrated that maximal multipartite entanglement generation is not a fine-tuned property of a specific dynamical protocol. Instead, it emerges as a generic feature of the two-walker quantum walk provided that a finite amount of initial coin entanglement is present. Except for a narrow neighborhood of the Pauli-$X$ coin, the multipartite entanglement remains close to its theoretical upper bound throughout the accessible parameter space.\\
\indent From a broader perspective, these results indicate that two-walker discrete-time quantum walks constitute an efficient mechanism for converting initially localized bipartite quantum correlations into robust genuine multipartite entanglement involving both the internal and spatial degrees of freedom. The contrasting behavior of bipartite and multipartite entanglement further highlights that multipartite quantum correlations cannot, in general, be inferred from individual bipartite measures alone, emphasizing the importance of simultaneously characterizing both types of entanglement in quantum-walk dynamics.\\
\indent The framework developed here can be extended in several directions. It would be interesting to investigate the influence of decoherence\cite{Kendon2004,Kendon2007,Alberti_2014}, interactions between walkers\cite{Ahlbrecht_2012,tg6f-v65l}, higher-dimensional lattices\cite{doi:10.1126/science.1218448,Zhou2024MultiParticleQW}, topological quantum walks\cite{PhysRevA.82.033429,Kitagawa2012,PhysRevB.92.045424}, and many-walker systems\cite{doi:10.1126/science.1229957,Zhou2024MultiParticleQW} on the generation and redistribution of multipartite entanglement. Such studies may provide further insight into the role of quantum walks as controllable platforms for engineering multipartite quantum correlations and for developing quantum information-processing protocols based on coherent quantum transport.

\begin{acknowledgments} 
S.H. acknowledges the financial support from the Council of Scientific and Industrial Research (CSIR), India, through the award of Junior Research Fellowship (JRF).
\end{acknowledgments}

\appendix
\section{Maximum Possible Value of the Generalized Geometric Measure}
\label{Max_GGM}
In this appendix, we derive the maximum possible value of the generalized geometric measure (GGM) for the four-partite quantum system considered in the present work. The total Hilbert space is
\begin{equation}
\mathcal H=\mathcal H_{C_A}\otimes\mathcal H_{P_A}\otimes\mathcal H_{C_B}\otimes\mathcal H_{P_B},
\end{equation}
where the coin spaces are two-dimensional and the position spaces are $N$-dimensional.\\
For an arbitrary pure state $\ket{\Psi}$, the GGM is defined as
\begin{equation}
GGM(\ket{\Psi})=1-\max_{\mathcal A:\mathcal B}\lambda_{\mathcal A:\mathcal B}^{2},
\label{eq:GGM_app}
\end{equation}
where $\lambda_{\mathcal A:\mathcal B}$ denotes the largest Schmidt coefficient corresponding to the bipartition $\mathcal A:\mathcal B$.
For any bipartition with Schmidt rank $r$, the Schmidt coefficients satisfy
\begin{equation}
\sum_{i=1}^{r}\lambda_i^2=1.
\end{equation}
Consequently, the largest Schmidt coefficient obeys
\begin{equation}
\lambda_{\max}^2 \ge \frac{1}{r}.
\end{equation}
\\Consider the bipartition
\begin{equation}
C_A:(P_AC_BP_B).
\end{equation}
Since $\dim(\mathcal H_{C_A})=2$, the Schmidt rank across this bipartition satisfies $r\le2$, implying
\begin{equation}
\lambda_{\max}^2\ge\frac12.
\end{equation}
\\Because Eq.~(\ref{eq:GGM_app}) involves the maximum Schmidt coefficient over all possible bipartitions, we necessarily have
\begin{equation}
\max_{\mathcal A:\mathcal B}
\lambda_{\mathcal A:\mathcal B}^{2}
\ge
\frac12.
\end{equation}
Therefore,
\begin{equation}
GGM\le
1-\frac12
=
\frac12.
\end{equation}
Hence the generalized geometric measure of the two-walker quantum-walk state is bounded above by
\begin{equation}
\boxed{
GGM_{\max}
=
\frac12 }.
\end{equation}

\section{Analytical Derivation of the Periodic Dynamics for $N=4$}
\label{Analytical_N=4}

In this appendix, we analytically investigate the evolution of the two-walker discrete-time quantum walk on a four-site cycle with periodic boundary conditions. We demonstrate that the periodic behavior observed in the main text is independent of the choice of the initial coin state and arises from the structure of the walk operator itself.

\subsection{Separable initial coin state}

We first consider the initial state

\begin{equation}
|\Psi(0)\rangle_{AB}
=
|01\rangle_{C_AC_B}
\otimes
|22\rangle_{P_AP_B}.
\label{eq:Bsep0}
\end{equation}

Successive applications of the evolution operator

\begin{equation}
U_{AB}
=
S_{AB}
\Big[
(H\otimes H)
\otimes
(I_{P_A}\otimes I_{P_B})
\Big]
\end{equation}

yield

\begin{align}
|\Psi(1)\rangle_{AB}
&=
\frac{1}{2}
\Big(
|0033\rangle
-|0131\rangle\nonumber\\
&\qquad\qquad
+
|1013\rangle
-|1111\rangle
\Big)_{C_AC_BP_AP_B}
\label{eq:Bsep1}\\
|\Psi(2)\rangle_{AB}
&=
\frac{1}{2}
\Big(
|{-}{+}00\rangle
-|{-}{-}02\rangle\nonumber\\
&\qquad\qquad
+|{+}{+}20\rangle
-|{+}{-}22\rangle
\Big)_{C_AC_BP_AP_B},\\
|\Psi(3)\rangle_{AB}
&=
|{+}{-}31\rangle_{C_AC_BP_AP_B},\\
|\Psi(4)\rangle_{AB}
&=
|01\rangle_{C_AC_B}
\otimes
|00\rangle_{P_AP_B},
\label{eq:Bsep4}\\
|\Psi(5)\rangle_{AB}
&=
\frac{1}{2}
\Big(
|0011\rangle
-|0113\rangle\nonumber\\
&\qquad\qquad
+|1031\rangle
-|1133\rangle
\Big)_{C_AC_BP_AP_B}.
\label{eq:Bsep5}
\end{align}
where $\ket+=(\ket0+\ket1)/\sqrt2$ and $\ket-=(\ket0-\ket1)/\sqrt2$.\\
\indent Comparing Eq.~(\ref{eq:Bsep0}) with (\ref{eq:Bsep4}) and  Eq.~(\ref{eq:Bsep1}) with (\ref{eq:Bsep5}), we observe that the coin state is exactly restored after four time steps. From Eq.~(\ref{eq:Bsep0}) and (\ref{eq:Bsep4}), we can see that, 
\begin{equation}
|01\rangle_{C_AC_B}
\longrightarrow
|01\rangle_{C_AC_B},
\end{equation}

while the position state undergoes a translation by two lattice sites,

\begin{equation}
|22\rangle_{P_AP_B}
\longrightarrow
|00\rangle_{P_AP_B}.
\end{equation}

\subsection{Bell-singlet initial coin state}

We next consider the Bell-singlet initial state

\begin{equation}
|\Psi(0)\rangle_{AB}
=
|\psi_3\rangle_{C_AC_B}
\otimes
|22\rangle_{P_AP_B},
\label{eq:Bpsi0}
\end{equation}

where

\begin{equation}
|\psi_3\rangle_{C_AC_B}
=
\frac{1}{\sqrt{2}}
\left(
|01\rangle-|10\rangle
\right)_{C_AC_B}.
\end{equation}

The evolution proceeds as

\begin{align}
|\Psi(1)\rangle_{AB}
&=
\frac{1}{\sqrt2}
\left(
-|0131\rangle+|1013\rangle
\right)_{C_AC_BP_AP_B}
\label{eq:Bpsi1},\\
|\Psi(2)\rangle_{AB}
&=
\frac{1}{2}
\Big[
|\phi^{+}\rangle_{C_AC_B}
(-|02\rangle+|20\rangle)_{P_AP_B}
\nonumber\\
&\qquad\qquad
+
|\psi^{-}\rangle_{C_AC_B}
(|00\rangle+|22\rangle)_{P_AP_B}
\Big],\\
|\Psi(3)\rangle_{AB}
&=
\frac{1}{\sqrt2}
\left(
-|{-}{+}13\rangle
+
|{+}{-}31\rangle
\right)_{C_AC_BP_AP_B},
\label{eq:Bpsi3}\\
|\Psi(4)\rangle_{AB}
&=
|\psi^{-}\rangle_{C_AC_B}
\otimes
|00\rangle_{P_AP_B},
\label{eq:Bpsi4}\\
|\Psi(5)\rangle_{AB}
&=\frac{1}{\sqrt2}
\left(
-|0113\rangle+|1031\rangle
\right)_{C_AC_BP_AP_B}.
\label{eq:Bpsi5}
\end{align}
\\where the state $\ket{\Phi^{+}}=(\ket{00}+\ket{11})/\sqrt{2}$ and $\ket{\Psi^{-}}=(\ket{01}-\ket{10})/\sqrt{2}$ are the usual Bell states while $\ket+=(\ket0+\ket1)/\sqrt2$ and $\ket-=(\ket0-\ket1)/\sqrt2$ are eigenstates of pauli-X matrix.

Comparing Eq.~(\ref{eq:Bpsi0}) with (\ref{eq:Bpsi4}) and  Eq.~(\ref{eq:Bpsi1}) with (\ref{eq:Bpsi5}), we again observe that the coin state is exactly restored after four time steps. From Eq.~(\ref{eq:Bpsi0}) and Eq. (\ref{eq:Bpsi4}), we can see that,
\begin{equation}
|\psi^{-}\rangle_{C_AC_B}
\longrightarrow
|\psi^{-}\rangle_{C_AC_B},
\end{equation}
while the position state is translated by two lattice sites,
\begin{equation}
|22\rangle_{P_AP_B}
\longrightarrow
|00\rangle_{P_AP_B}.
\end{equation}
\subsection{General periodicity}
The above calculations show that the four-step recurrence is independent of whether the initial coin state is separable or entangled. In both cases, the coin degrees of freedom are exactly restored after four time steps, whereas the position degrees of freedom are translated by two lattice sites modulo four.\\
More generally, after every four time steps the action of the evolution operator can be expressed as
\begin{equation}
|i,j,k,l\rangle_{C_AC_BP_AP_B}
\longrightarrow
|i,j,k+2,l+2\rangle_{C_AC_BP_AP_B} 
(\mathrm{mod}\;4).
\label{eq:Btranslation}
\end{equation}
\\
Hence, the coin degrees of freedom are unchanged, whereas each walker is translated by two lattice sites modulo four. Consequently,
\begin{equation}
|0\rangle_{P_r}
\leftrightarrow
|2\rangle_{P_r},
\qquad
|1\rangle_{P_r}
\leftrightarrow
|3\rangle_{p_r},
\qquad r=A,B,
\end{equation}
after every four time steps.\\
Thus, the periodic dynamics on a four-site cycle is characterized by an exact four-step recurrence in the coin sector together with a deterministic translation of the position sector by two lattice sites. This structure underlies the periodic entanglement dynamics discussed in the main text.
\section{Special Dynamics for the Pauli-$X$ Coin}
\label{app:Xcoin}
As discussed in Sec.\ref{subsection5c}, the robustness analysis reveals a pronounced suppression of the final-time GGM near $\epsilon_2=\pi/4$. Since this value corresponds to the Pauli-$X$ coin,

\begin{equation}
O_{C_A}(\epsilon_2=\pi/4)
=
\begin{pmatrix}
0 &
1 \\
1 &
0
\end{pmatrix}=X=O_{C_B}(\epsilon_2=\pi/4)
\label{eq:coin_X},
\end{equation}

we have examined the dynamics in this special case separately.
\begin{figure}[h]
    \centering
    \includegraphics[width=1.0\linewidth]{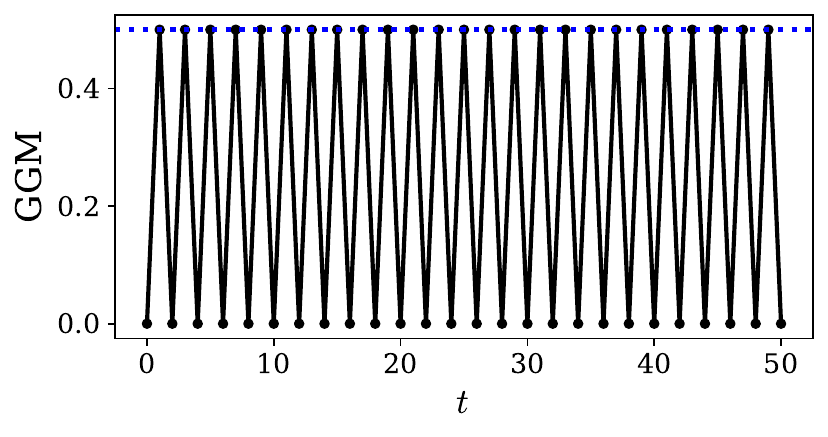}
    \caption{Time evolution of the generalized geometric measure (GGM) for a two-walker DTQW in the open-boundary regime ($N=101$, $T=50$) with the initial Bell-singlet state and $\epsilon_2=\pi/4$, corresponding to the Pauli-$X$ coin.}
    \label{fig:Xcoin_GGM}
\end{figure}
Fig.~\ref{fig:Xcoin_GGM} shows the time evolution of the generalized geometric measure for $\epsilon_2=\pi/4$. Unlike the generic behavior observed for other coin operators, the GGM exhibits a perfectly periodic oscillation between its minimum and maximum allowed values. Specifically,

\begin{equation}
GGM(t)=
\begin{cases}
0, & t \ \text{even},\\
\frac12, & t \ \text{odd},
\end{cases}
\end{equation}

\noindent throughout the evolution.

Thus, the system alternates exactly between states possessing no genuine multipartite entanglement and states attaining the maximum possible GGM allowed by the Hilbert-space structure. The absence of intermediate values indicates that the Pauli-$X$ coin generates a highly constrained evolution in which multipartite entanglement is periodically created and completely destroyed at successive time steps.

This behavior explains the sharp dip observed in the robustness analysis of Fig.\ref{fig:eps2_variation}. Since the final-time GGM depends on the parity of the chosen evolution time, the Pauli$X$ coin constitutes an exceptional point in parameter space where the long-time multipartite-entanglement generation mechanism becomes fundamentally different from that of generic coin operators.
\bibliographystyle{apsrev4-2}
\bibliography{references1}
\end{document}